\PassOptionsToPackage{unicode}{hyperref}
\PassOptionsToPackage{hyphens}{url}
\PassOptionsToPackage{dvipsnames,svgnames,x11names}{xcolor}
\documentclass[
  12pt]{article}

\usepackage{amsmath,amssymb}
\usepackage{iftex}
\ifPDFTeX
  \usepackage[T1]{fontenc}
  \usepackage[utf8]{inputenc}
  \usepackage{textcomp} 
\else 
  \usepackage{unicode-math}
  \defaultfontfeatures{Scale=MatchLowercase}
  \defaultfontfeatures[\rmfamily]{Ligatures=TeX,Scale=1}
\fi
\usepackage{lmodern}
\ifPDFTeX\else  
\fi
\IfFileExists{upquote.sty}{\usepackage{upquote}}{}
\IfFileExists{microtype.sty}{
  \usepackage[]{microtype}
  \UseMicrotypeSet[protrusion]{basicmath} 
}{}
\makeatletter
\@ifundefined{KOMAClassName}{
  \IfFileExists{parskip.sty}{%
    \usepackage{parskip}
  }{
    \setlength{\parindent}{0pt}
    \setlength{\parskip}{6pt plus 2pt minus 1pt}}
}{
  \KOMAoptions{parskip=half}}
\makeatother
\usepackage{xcolor}
\makeatletter
\ifx\paragraph\undefined\else
  \let\oldparagraph\paragraph
  \renewcommand{\paragraph}{
    \@ifstar
      \xxxParagraphStar
      \xxxParagraphNoStar
  }
  \newcommand{\xxxParagraphStar}[1]{\oldparagraph*{#1}\mbox{}}
  \newcommand{\xxxParagraphNoStar}[1]{\oldparagraph{#1}\mbox{}}
\fi
\ifx\subparagraph\undefined\else
  \let\oldsubparagraph\subparagraph
  \renewcommand{\subparagraph}{
    \@ifstar
      \xxxSubParagraphStar
      \xxxSubParagraphNoStar
  }
  \newcommand{\xxxSubParagraphStar}[1]{\oldsubparagraph*{#1}\mbox{}}
  \newcommand{\xxxSubParagraphNoStar}[1]{\oldsubparagraph{#1}\mbox{}}
\fi
\makeatother

\usepackage{longtable,booktabs,array}
\usepackage{calc} 
\usepackage{etoolbox}
\makeatletter
\patchcmd\longtable{\par}{\if@noskipsec\mbox{}\fi\par}{}{}
\makeatother
\IfFileExists{footnotehyper.sty}{\usepackage{footnotehyper}}{\usepackage{footnote}}
\makesavenoteenv{longtable}
\usepackage{graphicx}
\makeatletter
\def\maxwidth{\ifdim\Gin@nat@width>\linewidth\linewidth\else\Gin@nat@width\fi}
\def\maxheight{\ifdim\Gin@nat@height>\textheight\textheight\else\Gin@nat@height\fi}
\makeatother
\setkeys{Gin}{width=\maxwidth,height=\maxheight,keepaspectratio}
\makeatletter
\def\fps@figure{htbp}
\makeatother

\makeatletter
\@ifpackageloaded{caption}{}{\usepackage{caption}}
\AtBeginDocument{%
\ifdefined\contentsname
  \renewcommand*\contentsname{Table of contents}
\else
  \newcommand\contentsname{Table of contents}
\fi
\ifdefined\listfigurename
  \renewcommand*\listfigurename{List of Figures}
\else
  \newcommand\listfigurename{List of Figures}
\fi
\ifdefined\listtablename
  \renewcommand*\listtablename{List of Tables}
\else
  \newcommand\listtablename{List of Tables}
\fi
\ifdefined\figurename
  \renewcommand*\figurename{Figure}
\else
  \newcommand\figurename{Figure}
\fi
\ifdefined\tablename
  \renewcommand*\tablename{Table}
\else
  \newcommand\tablename{Table}
\fi
}
\@ifpackageloaded{float}{}{\usepackage{float}}
\floatstyle{ruled}
\@ifundefined{c@chapter}{\newfloat{codelisting}{h}{lop}}{\newfloat{codelisting}{h}{lop}[chapter]}
\floatname{codelisting}{Listing}

\makeatother
\makeatletter
\@ifpackageloaded{caption}{}{\usepackage{caption}}
\@ifpackageloaded{subcaption}{}{\usepackage{subcaption}}
\makeatother

\ifLuaTeX
  \usepackage{selnolig}  
\fi
\usepackage[]{natbib}
\usepackage{bookmark}

\IfFileExists{xurl.sty}{\usepackage{xurl}}{} 
\hypersetup{
  pdftitle={Title},
  pdfauthor={Author 1; Author 2},
  pdfkeywords={3 to 6 keywords, that do not appear in the title},
  colorlinks=true,
  linkcolor={blue},
  filecolor={Maroon},
  citecolor={Blue},
  urlcolor={Blue},
  pdfcreator={LaTeX via pandoc}}

\usepackage{algorithm}
\usepackage{algpseudocode}
\usepackage{tikz}

\usepackage{array}
\usepackage{textcomp}
\usepackage{stfloats}
\usepackage{url}
\usepackage{verbatim}
\usepackage{etoolbox}
\usepackage{graphicx}
\usepackage{enumerate}
\usepackage{bm}
\usepackage{longtable} 
\usepackage{setspace}

\usepackage{dutchcal}
\usepackage{enumerate}

\usepackage{graphicx}
\usepackage{booktabs}
\usepackage{threeparttable}
\usepackage{xcolor}
\usepackage{multirow} 
\usepackage{pifont}
\usepackage{color}

\usepackage[thinc]{esdiff}
\usepackage{amsthm}

\theoremstyle{plain}
\newtheorem{theorem}{Theorem}

\newtheorem{proposition}{Proposition}

\theoremstyle{definition}
\newtheorem{definition}{Definition}
\newtheorem{assumption}{Assumption}

\theoremstyle{remark}
\newtheorem{remark}{Remark}

\makeatletter
\def\thm@space@setup{%
  \thm@preskip=3pt
  \thm@postskip=3pt
}
\makeatother

\newcommand{\cX}{\mathcal{X}}
\newcommand{\cY}{\mathcal{Y}}
\newcommand{\bbR}{\mathbb{R}}
\newcommand{\bbE}{\mathbb{E}}
\newcommand{\cL}{\mathcal{L}}
\newcommand{\cF}{\mathcal{F}}
\newcommand{\cT}{\mathcal{T}}
\newcommand{\Bf}{\mathbf{f}}
\newcommand{\bx}{\mathbf{x}}
\newcommand{\bu}{\mathbf{u}}
\newcommand{\bp}{\mathbf{p}}
\newcommand{\bq}{\mathbf{q}}
\newcommand{\bC}{\mathbf{C}}
\newcommand{\bW}{\mathbf{W}}
\newcommand{\bI}{\mathbf{I}}

\newcommand{\bR}{\mathbf{R}}
\newcommand{\btheta}{\symbfit{\theta}}
\newcommand{\bgamma}{\symbfit{\gamma}}

\newcommand{\Halmos}{\hfill\ensuremath{\square}}

\newcommand{\anon}{1}

\makeatletter

\renewcommand\section{\@startsection{section}{1}{\z@}%
  {-1.4ex \@plus -.3ex \@minus -.2ex}%
  {0.5ex \@plus .1ex}%
  {\normalfont\Large\bfseries}}

\renewcommand\subsection{\@startsection{subsection}{2}{\z@}%
  {-1.1ex \@plus -.2ex \@minus -.2ex}%
  {0.4ex \@plus .1ex}%
  {\normalfont\large\bfseries}}

\renewcommand\subsubsection{\@startsection{subsubsection}{3}{\z@}%
  {-0.9ex \@plus -.2ex \@minus -.1ex}%
  {0.3ex \@plus .1ex}%
  {\normalfont\normalsize\bfseries}}

\makeatother

\begin{document}

\def\spacingset#1{\renewcommand{\baselinestretch}%
{#1}\small\normalsize} \spacingset{1}


\if1\anon
{
  \title{\bf A Scalable Operator-Valued Kernel Learning Framework for Continuous Multi-Way Regression}
  \author{Haijie Xu\\
    Department of Industrial Engineering, Tsinghua University\\
    Qiang Zhang\\
    School of Management and Economics,\\ The Chinese University of Hong Kong, Shenzhen \\
    Chen Zhang \\
    Department of Industrial Engineering, Tsinghua University}
  \maketitle
} \fi

\if0\anon
{
  \bigskip
  \bigskip
  \bigskip
  \begin{center}
    {\LARGE\bf A Scalable Operator-Valued Kernel Learning Framework for Continuous Multi-Way Regression}
\end{center}
  \medskip
} \fi

\bigskip
\begin{abstract}
Continuous multi-way data arise in applications such as climate modeling, neuroimaging, and multimedia analytics, where both covariates and responses are indexed over structured spatial or temporal continuous domains. Existing methods discretize such data as fixed-grid arrays and use tensor regression methods. Functional regression instead models observations directly over continuous domains. Yet existing methods are primarily developed for one-way predictors and responses and rely heavily on finite-dimensional basis representations. To fill this gap, we develop FRCOMA, a scalable operator-valued kernel learning framework for nonlinear regression with heterogeneous continuous multi-way covariates and responses. The proposed method uses separable operator-valued kernels to decouple input similarity from output functional dependence, leading to a Kronecker-structured kernel system that can be solved through spectral computation. We further introduce sparse kernel weights to identify informative covariates and develop a block coordinate descent algorithm for joint estimation. We establish a representer theorem, support recovery guarantees, and algorithmic convergence results under regularity conditions. Numerical studies and applications to movie-trailer engagement prediction and climate-field modeling show that FRCOMA improves prediction and variable selection relative to functional and tensor regression baselines, especially when continuity and heterogeneous multi-way structure are important.
\end{abstract}

\noindent%
\vfill

\newpage
\spacingset{1.3} 

\section{INTRODUCTION}
\label{sec:intro}

Modern sensing technologies increasingly generate data indexed over multiple continuous spatial or temporal domains, leading to what is commonly referred to as \textit{continuous multi-way data}. Such data arise in applications including climate modeling, neuroimaging, chemometrics, and hyperspectral imaging \citep{hersbach2020era5,acar2007multiway,anzardi2021processing,renard2006multiway}. A fundamental computational task is to learn nonlinear mappings between high-dimensional structured objects, where both predictors and responses are indexed over continuous multi-dimensional domains.

Existing approaches typically discretize continuous objects into finite tensors and then apply tensor regression or tensor decomposition methods \citep{lock2018tensor,liu2020low,liu2020smooth}. Although computationally convenient, these methods rely on finite-dimensional approximations and remain dependent on predefined grids. As the number of modes or grid resolution increases, the computational cost grows rapidly. Moreover, irregularly sampled observations often require additional preprocessing or interpolation, which may introduce approximation errors. Extensions to multiple tensor-valued covariates have been developed \citep{gahrooei2021multiple}, but they still follow the discretization-based paradigm.

Functional data analysis (FDA) provides an alternative by modeling observations as continuous functions. Classical function-on-function regression methods based on basis expansions and functional principal component analysis (FPCA) have been widely studied \citep{yao2005functional,wang2014linear,luo2016functional,luo2017function}, with extensions to variable selection \citep{cai2021variable}. However, these approaches usually rely on low-dimensional basis representations and are primarily designed for one-dimensional functional data, limiting their scalability to heterogeneous multi-way functional structures. Recent studies on multi-way functional representations \citep{han2024guaranteed,larsen2024tensor} partially bridge FDA and tensor analysis, but mainly focus on representation or decomposition rather than regression with multiple heterogeneous continuous multi-way covariates and responses.

Kernel methods provide a natural framework for nonlinear learning in functional spaces. Reproducing kernel Hilbert spaces (RKHS) and operator-valued kernels have been developed for functional response regression \citep{ferraty2006,baillo2009,ferraty2012,lian2007nonlinear,kadri2016operator}. However, existing kernel-based functional regression methods are generally limited to single functional predictors or one-dimensional functional objects and do not directly address nonlinear regression with heterogeneous continuous multi-way inputs and outputs.

Motivated by these limitations, we consider a regression problem where both predictors and responses are continuous multi-way functional objects. Specifically, the response $y_i$ is a $d_y$-way continuous function, and the covariate vector $\bx_i=(x_i^{(1)},\ldots,x_i^{(p)})$ consists of multiple functional covariates, where each $x_i^{(l)}$ may be a $d_l$-way continuous function with heterogeneous dimensions across covariates. We formulate the relationship as
\begin{equation}
\label{eq:F2F}
    y_i=f(\bx_i)+\epsilon_i,\quad i=1,\dots,n,
\end{equation}
where $f$ is an unknown nonlinear operator mapping heterogeneous continuous multi-way covariates to a continuous multi-way response. Unlike tensor-on-tensor regression, our formulation directly models objects over continuous domains rather than fixed-grid arrays. It also generalizes classical function-on-function regression, which corresponds to the special case $d_y=d_l=1$ for all $l$. In high-dimensional settings, the framework should additionally identify informative covariates and remove irrelevant ones.

We propose a scalable operator-valued kernel learning framework, termed \textbf{FRCOMA}, for nonlinear regression with continuous multi-way functional data. The key idea is to construct separable operator-valued kernels that decouple input similarity from output functional dependence, yielding a Kronecker-structured kernel representation that enables efficient spectral computation. Compared with discretization-based or basis-expansion approaches, FRCOMA preserves the continuous structure of functional objects while achieving scalable computation through structured operator decomposition. Furthermore, we introduce sparse kernel weights for automatic variable selection and develop a block coordinate descent algorithm for joint estimation. The main contributions are summarized as follows: 1) We develop a scalable operator-valued kernel framework for nonlinear regression with heterogeneous continuous multi-way covariates and responses. 2) We derive a Kronecker-structured spectral computation scheme that enables efficient inversion of the operator-valued kernel system. 3) We develop a block coordinate descent algorithm for joint estimation and sparse variable selection in high-dimensional settings. 4) We establish theoretical guarantees for support recovery and algorithmic convergence.

The remainder of this article is organized as follows. Section~\ref{sec:literature} reviews the related literature. Section~\ref{sec:method} introduces the proposed FRCOMA framework and derives its closed-form solution. Section~\ref{sec:computation} extends FRCOMA with sparse regularization and presents a BCD algorithm. Section~\ref{sec:theory} establishes the theoretical properties. Section~\ref{sec:simulation} presents numerical experiments, and Section~\ref{sec:case} demonstrates applications to two real-world datasets. Finally, Section~\ref{sec:conslusion} concludes the article.

\section{RELATED WORK}
\label{sec:literature}

\subsection{Tensor Regression for Discretized Multi-Way Data}

Tensor-on-tensor regression studies regression problems in which both predictors and responses are represented as finite-dimensional tensors. An early contribution by \citet{hoff2015multilinear} models tensor-valued responses and covariates through mode-wise regression structures, assuming compatible modes between the covariate and response tensors. To relax this restriction, \citet{lock2018tensor} proposes a low-rank tensor regression framework in which their relationship is represented by a high-dimensional coefficient tensor with a CP decomposition, substantially reducing the number of free parameters and improving interpretability. Building on this paradigm, \citet{liu2020low} adopts tensor-train decomposition, while \citet{liu2020smooth} uses tensor-ring decomposition together with group sparsity to control model complexity.

Tensor regression has subsequently been extended to other structured settings. For example, \citet{llosa2022reduced} develops an analysis-of-variance framework to identify regions of the covariate tensor associated with the response, while \citet{billio2023bayesian} proposes a Bayesian dynamic tensor autoregression model for time-varying tensor data. These approaches provide useful finite-dimensional modeling tools, with computational efficiency largely achieved through low-rank tensor structures. However, they remain tied to fixed-grid representations and do not directly model the underlying continuous domains from which many multi-way observations arise.

Nonparametric and kernel-based approaches have also been developed for discrete tensor data. \citet{zhao2013kernel} extends partial least squares using a multi-mode product kernel to capture nonlinear relationships between tensor-valued covariates and responses. Gaussian process models provide another direction: \citet{yu2018tensor} proposes a Gaussian process tensor-on-tensor regression model with a multilinear kernel and establishes its connection to low-rank tensor regression, while \citet{lei2024scalable} combines CP decomposition of the coefficient tensor with Gaussian process modeling of the resulting matrix factors for scalable spatiotemporal regression.

Most tensor-on-tensor regression methods consider a single tensor-valued covariate, whereas settings with multiple tensor-valued covariates have received comparatively less attention \citep{gahrooei2021multiple,hou2016common}. \citet{gahrooei2021multiple} constructs a separate coefficient tensor for each covariate and uses Tucker decomposition for dimension reduction, together with a shared basis structure in the response space. \citet{hou2016common} extends kernel-based tensor partial least squares to multiple covariates through a multi-kernel learning formulation. Despite these advances, such methods still treat observations as discrete arrays. Their computation and prediction therefore depend on the selected grids, and irregularly sampled or misaligned continuous data generally require preprocessing, resampling, or interpolation.

\subsection{Functional Regression and Continuous Multi-Way Representations}

Functional data analysis provides a complementary perspective by modeling observations as realizations of underlying continuous functions. In the classical function-on-function setting, both predictors and responses are one-dimensional functions. Representative approaches include functional linear regression for longitudinal functional data \citep{yao2005functional}, functional additive models \citep{muller2008functional}, and functional linear mixed models \citep{wang2014linear}. These methods commonly use functional principal component analysis or other basis expansions to project infinite-dimensional functions onto finite-dimensional coefficient spaces.

The presence of multiple functional covariates introduces additional challenges due to increased dimensionality and dependence among covariates. \citet{chiou2014multivariate} develops a multivariate FPCA framework for jointly representing multiple functional objects in a reduced-dimensional space. \citet{luo2016functional} uses wavelet representations for function-on-function regression, while \citet{luo2017function} formulates a penalized functional eigenvalue problem for estimating coefficient functions. Although effective for one-dimensional functional regression, these approaches do not naturally extend to heterogeneous multi-way functional covariates and responses.

Variable selection has also been studied in functional regression. \citet{hong2011inference} uses a Lasso penalty to select relevant functional predictors, but assigns a scalar coefficient to each covariate, limiting its ability to capture full function-to-function dependence. \citet{cai2021variable} addresses this issue by modeling each regression coefficient as a two-dimensional function and applying a group SCAD penalty for function-level variable selection. Nevertheless, existing variable-selection methods are primarily designed for one-dimensional functional predictors, namely, $d_l=1$ for all $l=1,\ldots,p$, and do not directly accommodate heterogeneous continuous multi-way covariates.

Recent work has begun to bridge tensor methods and functional data analysis through continuous multi-way representations. \citet{han2024guaranteed} develops a functional tensor singular value decomposition with theoretical guarantees, while \citet{larsen2024tensor} studies tensor decomposition in RKHS-based settings for smooth and misaligned data. These studies represent important progress in modeling continuous multi-way structures, but their primary focus is representation or decomposition rather than regression from multiple heterogeneous continuous multi-way covariates to a continuous multi-way response. A general computational framework for nonlinear continuous multi-way regression is therefore still lacking.

\subsection{Kernel-Based Operator-Valued Learning}

Kernel methods provide a general framework for nonlinear learning through reproducing kernel Hilbert spaces (RKHS). In the scalar-valued setting, regularized empirical risk minimization admits finite-dimensional solutions through the representer theorem \citep{aronszajn1950theory,scholkopf2001generalized,wahba1990spline}.

RKHS theory has subsequently been extended to structured-output learning, including vector-valued and multi-task settings, where kernels are modeled as bounded linear operators between Hilbert spaces to capture dependencies among output components \citep{micchelli2005learning,alvarez2012kernels,evgeniou2005learning}. Tensor-valued kernel methods further extend structured learning to finite-dimensional multi-way outputs by modeling interactions across tensor modes through separable or low-rank kernel structures \citep{signoretto2011tensor,zhou2013tensor,huang2012tensor}.

A related but distinct line of work considers operator-valued kernels for learning mappings between infinite-dimensional functional spaces. Unlike tensor-valued kernels, which typically rely on finite-dimensional tensor representations, operator-valued kernels define mappings directly between Hilbert spaces and allow functional responses to be modeled without discretizing the output space. In this setting, learning amounts to estimating operator-valued regression functions, and the resulting optimization problems involve block kernel systems in which each kernel evaluation is itself a bounded linear operator. Such methods have been studied in functional and structured prediction settings \citep{lian2007nonlinear,kadri2016operator}. However, the resulting systems are computationally demanding and scale poorly with both sample size and functional resolution.

When multiple covariates are present, an additional issue is how to construct or select among covariate-specific kernel representations. Multiple kernel learning (MKL) addresses this problem by combining candidate kernels through data-driven weights, thereby enabling adaptive selection among RKHS-induced representations \citep{lanckriet2004learning,bach2004multiple,rakotomamonjy2008simplemkl}. In our setting, this corresponds to kernel-level selection over covariate-specific functional RKHSs, with each covariate contributing its own kernel component. Nevertheless, existing MKL methods are mainly developed for scalar or finite-dimensional outputs and do not directly extend to scalable operator-valued regression with heterogeneous continuous multi-way functional covariates and responses.

\section{Scalable Operator-Valued Kernel Regression for Continuous Multi-Way Data}
\label{sec:method}
\subsection{Operator-Valued RKHS Formulation}
\label{subsec:model}

We now provide a formal operator-valued RKHS formulation of the regression model introduced in Section~\ref{sec:intro}. Consider $n$ observations $\{(y_i,\bx_i)\}_{i=1}^n$, where
\begin{equation}
\label{eq:F2F}
    y_i = f(\bx_i) + \epsilon_i,\quad i=1,\dots,n,
\end{equation}
and $f$ is an unknown nonlinear mapping to be estimated.

The response $y_i$ is a continuous $d_y$-way functional object defined over the compact domain $\Omega_y \triangleq \cT_1^y\times\cT_2^y\times\cdots\times \cT_{d_y}^y \subset \bbR^{d_y}$, where each $\cT_k^y\subset\bbR$ is a compact interval. We use $\cY \triangleq L^2(\Omega_y)$ as the ambient Hilbert space for the response, so that $y_i\in\cY$. For any $u,v\in\cY$, the inner product and norm are defined as $\langle u,v\rangle_{\cY}=\int_{\Omega_y}u(t)v(t)\,dt$ and 
$\|u\|_{\cY}=\sqrt{\langle u,u\rangle_{\cY}}$, respectively. The covariates are denoted by $\bx_i=(x_i^{(1)},\dots,x_i^{(p)})$, where each $x_i^{(l)}$ is a continuous $d_l$-way functional object defined over $\Omega_{x_l}\triangleq \cT_1^l\times\cT_2^l\times\cdots\times \cT_{d_l}^l \subset \bbR^{d_l}$, with each $\cT_k^l\subset\bbR$ compact. We write $\cX_l\triangleq L^2(\Omega_{x_l})$ and $x_i^{(l)}\in\cX_l$ for $l=1,\dots,p$, where $\cX_l$ is endowed with the analogous $L^2$ inner product and norm. This formulation allows the covariates to have heterogeneous intrinsic dimensions $d_l$. Without loss of generality, all domains are normalized to the unit hypercube, i.e., $\Omega_y=[0,1]^{d_y}$ and $\Omega_{x_l}=[0,1]^{d_l}$ for $l=1,\dots,p$. In practical applications, some covariates may be irrelevant for predicting the response. We refer to such covariates as non-informative variables and later incorporate variable selection mechanisms to identify and exclude them. The error term $\epsilon_i$ is assumed to be a mean-zero $\cY$-valued random element, independent across $i$ and independent of the covariates. Specifically, we assume $E(\epsilon_i)=0$ and $E\|\epsilon_i\|_{\cY}^2<\infty$. For theoretical convenience, we further allow Gaussian $\cY$-valued errors $\epsilon_i\sim N_{\cY}(0,\Gamma_\epsilon)$, where $\Gamma_\epsilon:\cY\to\cY$ is a self-adjoint, positive semidefinite covariance operator satisfying $\operatorname{tr}(\Gamma_\epsilon)<\infty$. 

Under this setting, the regression operator in Eq.~\eqref{eq:F2F} is a mapping $f:\cX_1\times\cdots\times\cX_p\to\cY$ between separable Hilbert spaces. We assume that $f$ admits the additive decomposition
\begin{equation}
\label{eq:general reg}
    y_i=\sum_{l=1}^p \alpha_l f_l(x_i^{(l)})+\epsilon_i,\quad i=1,\dots,n,
\end{equation}
where $\alpha_l>0$, $l=1,\dots,p$, are weight hyperparameters representing the contribution of each component $f_l(x_i^{(l)})$. Each component function $f_l$ is assumed to belong to a functional RKHS $\cF_l$ with an associated operator-valued kernel $K_l(\cdot,\cdot)$, defined formally as follows.

\begin{definition}[Operator-valued Kernel]
\label{def:op kernel}
   Denote the set of bounded linear operators from $\cY$ to $\cY$ as $\cL(\cY)$. An operator-valued kernel $K( \cdot,\cdot): \cX \times\cX \rightarrow \cL(\cY)$ is said to satisfy 
   
    (i) Hermitian: $\forall w,z \in \cX$, $K(w,z) = K(z,w)^{*}$, where the $^{*}$ denotes the adjoint operator; 
    
    (ii) Positive semi-definiteness: $K$ is Hermitian, and for every $r \in \mathbb{N}$ and all $\{(w_i,u_i)\}_{i = 1}^r\in \cX \times \cY$, the matrix with its $ij$-th element $\langle K(w_i,w_j)u_i, u_j \rangle_{\cY}$ is positive semidefinite.
\end{definition}
\begin{definition}[Functional RKHS]
\label{def:funRKHS}
    A Hilbert space $\cF$ of functions from $\cX$ to $\cY$ is called a functional RKHS if there is an operator-valued kernel $K( \cdot,\cdot): \cX \times\cX \rightarrow \cL(\cY)$ such that 
    
    (i) for every $w\in\cX$ and $g\in\cY$, the function 
$K(w,\cdot)g:\cX\to\cY$, defined by $z\mapsto K(w,z)g$, belongs to $\cF$;

    (ii) for every $F \in \cF, w \in \cX$ and $g \in \cY$, the reproducing property holds: $\langle F,K(w,\cdot)g \rangle_{\cF} = \langle F(w), g \rangle_{\cY}$.
\end{definition}

Based on the above definitions,  the regression model in Eq.~\eqref{eq:general reg} can be viewed as a linear combination of nonlinear mappings from multiple functional RKHSs. Although each $f_l$ captures nonlinear relationships between the covariate and response, the additive formulation helps control model complexity. This is particularly beneficial in high-dimensional settings, where fully nonparametric models often suffer from the curse of dimensionality, resulting in poor convergence rates. Additive models provide a favorable trade-off, maintaining flexibility while enabling faster convergence \citep{han2020additive}.

\begin{remark} \label{remark1} The choice of an operator-valued kernel is important for both statistical modeling and numerical computation. An improperly chosen kernel may lead to non-invertibility or numerical instability. To obtain a tractable operator-valued kernel system, we focus on separable kernels of the form \begin{equation} \label{eq:separable kernel} K_l(w,z)=g_l(w,z)T,\quad w,z\in\cX_l,\quad l=1,\dots,p, \end{equation} where $g_l:\cX_l\times\cX_l\to\bbR$ is a scalar-valued positive semidefinite kernel and $T\in\cL(\cY)$ is a common bounded, self-adjoint, positive semidefinite operator on $\cY$. This separable construction follows the operator-valued kernel framework in \cite{micchelli2004kernels,micchelli2005learning}. For example, $g_l$ can be chosen as a Gaussian kernel, $g_l(w,z)=\exp(-\|w-z\|_{\cX_l}^2/\sigma_g^2)$, where $\|\cdot\|_{\cX_l}$ denotes the norm on $\cX_l$ and $\sigma_g$ is the kernel bandwidth. The operator $T$ may be defined as an integral operator: for any $u\in\cY$, $(Tu)(t)=\int_{\Omega_y} k(s,t)u(s)\,ds$, where $k:\Omega_y\times\Omega_y\to\bbR$ is a scalar-valued positive semidefinite kernel. 
\end{remark} 

The separable construction in Remark~\ref{remark1} induces a structured operator-valued RKHS that plays a central role in the computational design of the proposed framework. In particular, it enables a separation between input similarity and output functional dependence, which leads to efficient kernel representations. We now study the corresponding regularity properties of the induced RKHS.

\begin{remark}[Continuity induced by separable operator-valued kernels]
\label{remark:continuity}
Under the separable construction in Remark~\ref{remark1}, each operator-valued kernel takes the form $K_l(w,z)=g_l(w,z)T$, where $g_l:\cX_l\times\cX_l\to\mathbb{R}$ is a scalar-valued kernel and $T\in\cL(\cY)$ is a bounded linear operator. Under the assumption that $g_l$ is locally bounded and continuous with respect to the norm topology on $\cX_l$, then the induced RKHS $\cF_l$ consists of functions that are continuous mappings from $\cX_l$ to $\cY$, i.e., $\cF_l \subset C(\cX_l,\cY)$. This implies that each component function $f_l$ in the additive model inherits continuity with respect to its corresponding functional covariate $x^{(l)}$, ensuring stability of the regression mapping under perturbations of the input functions. Consequently, the overall regression function $f=\sum_{l=1}^p \alpha_l f_l$ defines a continuous mapping from the product space $\cX_1\times\cdots\times\cX_p$ to $\cY$. A similar setting is also adopted in \cite{kadri2016operator}.
\end{remark}

\subsection{Representer-Based Estimation}
Given $n$ training samples, $(y_i,\bx_i) \in \cY \times \cX, i = 1,\dots,n$, where $\cX \triangleq \cX_1 \times \dots \times \cX_p$, we estimate the component functions $f_{l}, l = 1,\ldots, p$ by solving the following optimization problem:
\begin{equation}
\begin{aligned}
\label{eq:model 1}
    \hat{f}_{1},\ldots,\hat{f}_{p} = \mathop{\arg\min}\limits_{f_{l}\in \cF_l,  l = 1,\ldots, p} &\sum_{i = 1}^{n}||y_i - \sum_{l = 1}^p \alpha_l f_l(x_i^{(l)}) ||^2_{\cY}  
    + \lambda_1 \sum_{l = 1}^{p} \theta_l || f_{l}||_{\cF_l}^2,
\end{aligned}
\end{equation}
where $\cF_l$ denotes the functional RKHS associated with its operator-valued kernel $K_l$. The scalar $\lambda_1 > 0$ is a regularization parameter that controls the overall complexity of the model, and $\theta_l, \alpha_l >0, l = 1,\dots,p$ are the weight hyperparameters that balance the contribution and smoothness of each covariate. All hyperparameters $\alpha_l$, $\theta_l$ and $\lambda_1$ are assumed to be predetermined.  

\begin{theorem}[Representer Theorem]
\label{thm:representer}
    The optimized $\hat{f}_l$ in Eq.~\eqref{eq:model 1} for each $l = 1,\ldots, p$ admits the following representation:
    \begin{equation}
        \hat{f}_l(\cdot) = \sum_{j = 1}^n K_l(x_j^{(l)},\cdot)u_j^{(l)},
    \end{equation}
    where $u_j^{(l)} \in \cY$ for all $j\in{1,\ldots,n}$.
\end{theorem}

\begin{proof}{Proof}
    The proof of Theorem \ref{thm:representer} is shown in Appendix 1. \Halmos
\end{proof}
 Based on Theorem \ref{thm:representer}, we can reformulate the original problem in terms of the coefficients $\bu \in \cY^{n\times p}$ where its $jl$-th element is $u_j^{(l)}$. The optimization problem in Eq. (\ref{eq:model 1}) becomes:
\begin{equation}
\label{eq:model 2}
    \begin{aligned}
    \hat{\bu} = \mathop{\arg\min}\limits_{\bu \in \cY^{n\times p}} &\sum_{i = 1}^{n} ||y_i - \sum_{l = 1}^{p}\sum_{j = 1}^{n}\alpha_l K_l(x_j^{(l)},x_i^{(l)})u_j^{(l)} ||_{\cY}^2 + \lambda_1\sum_{l = 1}^{p} \theta_l \sum_{i,j = 1}^{n} \langle K_l(x_i^{(l)},x_j^{(l)})u_i^{(l)}, u_j^{(l)} \rangle_{\cY}.
    \end{aligned}
\end{equation}
For fixed positive $\alpha_l$ and $\theta_l$, the scaling of $\alpha_l$ and $\theta_l$ can be absorbed into the coefficient functions. Therefore, without loss of generality, we work under the normalized parameterization $\alpha_l=\theta_l$ for $l=1,\dots,p$. Equivalently, if $\alpha_l\ne\theta_l$, one may define transformed weights $\tilde\alpha_l=\tilde\theta_l=\alpha_l^2/\theta_l$, solve the corresponding problem to obtain $\tilde u_j^{(l)}$, and recover the original coefficients by $\hat u_j^{(l)}=\tilde u_j^{(l)}\alpha_l/\theta_l$. This normalization leads to the following parameter-sharing property.

\begin{proposition}
    \label{thm:only u}
    If $\alpha_l = \theta_l$ for all $l = 1,\ldots, p$, then the optimal coefficients in Eq. (\ref{eq:model 2}) satisfy $\hat{u}_i^{(1)} = \cdots = \hat{u}_i^{(p)}, \forall i = 1,\dots,n$.
\end{proposition}

\begin{proof}{Proof}
    The proof of Proposition \ref{thm:only u} is shown in Appendix 2. \Halmos
\end{proof}    

Based on Proposition \ref{thm:only u}, we can define a unified representation $u_i \triangleq u_i^{(l)}$ for all $l = 1,\dots,p$, since the optimal coefficients are shared across all covariates. This allows us to reduce the parameter space and rewrite Eq.~\eqref{eq:model 2} in a more compact form. Specifically, we reformulate the optimization problem as: 
\begin{equation}
\label{eq:model 4}
    \begin{aligned}
    \hat{\bu} = \mathop{\arg\min}\limits_{\bu \in \cY^{n}} &\sum_{i = 1}^{n} ||y_i - \sum_{l = 1}^{p}\sum_{j = 1}^{n}\theta_l K_l(x_j^{(l)},x_i^{(l)})u_j ||_{\cY}^2 &+ \lambda_1\sum_{l = 1}^{p} \theta_l \sum_{i,j = 1}^{n} \langle K_l(x_i^{(l)},x_j^{(l)})u_i, u_j \rangle_{\cY},
    \end{aligned}
\end{equation}
where $\bu=(u_1, \ldots, u_n)^T$ and $\theta_l, l = 1,\dots,p$ are the hyperparameters.

\subsection{Kronecker-Structured Spectral Computation}
For fixed $\btheta \triangleq (\theta_1,\dots,\theta_p)^T$, the coefficient system induced by Eq.~\eqref{eq:model 4} admits a closed-form spectral solution. Define the aggregated operator-valued kernel $K(\bx_i,\bx_j)=\sum_{l=1}^{p}\theta_l K_l(x_i^{(l)},x_j^{(l)})$, associated with $\btheta$. Since $\theta_l\ge0$ and each $K_l$ is positive semidefinite, $K(\cdot,\cdot):\cX\times\cX\to\cL(\cY)$ is also positive semidefinite. With this notation, Eq.~\eqref{eq:model 4} can be written as 
\begin{equation}
    \label{eq:solve u sim}
    \begin{aligned}
    \hat{\bu} = \mathop{\arg\min}\limits_{\bu \in \cY^n} & m_{\btheta}(\bu) 
    = \mathop{\arg\min}\limits_{\bu \in \cY^n} &\sum_{i = 1}^{n} ||y_i - \sum_{j = 1}^nK(\bx_j,\bx_i)u_j ||_{\cY}^2 &+ \lambda_1 \sum_{i,j = 1}^{n}\langle K(\bx_i,\bx_j)u_i,u_j \rangle_{\cY}.
    \end{aligned}
\end{equation}

This reformulation allows Eq.~\eqref{eq:solve u sim} to be interpreted as a standard functional regression problem in RKHS, as discussed in \cite{kadri2016operator}. Its optimality condition yields
\begin{equation}
\label{eq:big K}
    (\mathbf{K}_{\btheta} + \lambda_1\bI)\hat{\bu} = \mathbf{y},
\end{equation}
where $\mathbf{y} = (y_1,\dots, y_n)^T\in \cY^n$ and $\mathbf{K}_{\btheta} \in \cL(\cY)^{n \times n}$ is a block operator-valued kernel matrix, whose $ij$-th element is given by $K(\bx_i,\bx_j) \in \cL(\cY)$. Therefore, the primary computational task becomes solving $\hat{\bu}=(\mathbf{K}_{\btheta} + \lambda_1\bI)^{-1}\mathbf{y}$. 

Under the separable construction in Remark~\ref{remark1}, $K_l(x_i^{(l)},x_j^{(l)})=g_l(x_i^{(l)},x_j^{(l)})T$ with a common output operator $T$ across covariates. Therefore, the aggregated kernel can be written as $K(\bx_i,\bx_j)=g(\bx_i,\bx_j)T$, where $g(\bx_i,\bx_j)=\sum_{l=1}^{p}\theta_l g_l(x_i^{(l)},x_j^{(l)})$. Consequently, the block operator-valued kernel matrix has the Kronecker structure $\mathbf K_{\btheta}=\mathbf G_{\btheta}\otimes T$, where $\mathbf G_{\btheta}\in\bbR^{n\times n}$ has $(i,j)$th entry $g(\bx_i,\bx_j)$. This structure is the key computational device that allows the operator system in Eq.~\eqref{eq:big K} to be solved spectrally. To efficiently compute $(\mathbf{K}_{\btheta} + \lambda_1\bI)^{-1}\mathbf{y}$, we adopt spectral decomposition on $\mathbf{G}_{\btheta}$ and $T$. This allows us to exploit the Kronecker product structure of $\mathbf{K}_{\btheta}$ using the eigenvalues and eigenfunctions of both $\mathbf{G}_{\btheta}$ and $T$, and significantly reduce the computational cost. In particular, for the scalar kernel matrix $\mathbf{G}_{\btheta} \in \bbR^{n\times n}$, we can obtain a set of eigenvalues $\beta_b \in \bbR$ and the corresponding eigenvectors $\mathbf{v}_b\in\bbR^{n}$ for $b = 1,\dots,n$. According to Assumption \ref{ass:inverse operator} (to be introduced in Section~\ref{sec:theory}), we assume that $T$ is a finite-rank compact operator with rank $\kappa$. Thus according to Theorem 6 in \cite{kadri2016operator}, for the operator $T$, we can compute its eigenvalues $\delta_q \in \bbR$, and corresponding eigenfunctions $w_q\in \cY$ for $q = 1,\dots,\kappa$. Then $\mathbf v_b\otimes w_q$ is an eigenfunction of $\mathbf K_{\btheta}$ with eigenvalue $\beta_b\delta_q$. Hence the inverse of $\mathbf K_{\btheta}+\lambda_1\mathbf I$ can be evaluated by projecting $\mathbf y$ onto this tensor-product eigenbasis and scaling each component by $(\beta_b\delta_q+\lambda_1)^{-1}$. Algorithm~\ref{alg:solve u} summarizes this spectral solver.
\section{Sparse Operator-Valued Kernel Learning and Variable Selection}
\label{sec:computation}
\subsection{Sparse Kernel-Weight Formulation}
When the number of covariates 
$p$ is large, it is likely that not all the covariates are informative for predicting the response. To address this, we extend the model in Eq.~\eqref{eq:model 4} to incorporate variable selection. 
Specifically, instead of treating the kernel weight vector $\btheta$ as a fixed hyperparameter, we now treat it as a parameter to be estimated jointly with the regression coefficients $\bu$.

This leads to the following joint optimization problem: 
\begin{equation}
\label{eq:model 3}
    \begin{aligned}
        (\hat{\bu},\hat{\btheta}) = \mathop{\arg\min}\limits_{\bu \in \cY^n, \btheta \in \bbR_{+}^{p}} q(\bu,\btheta)  
        = &\mathop{\arg\min}\limits_{\bu \in \cY^n, \btheta \in \bbR_{+}^{p}} \sum_{i = 1}^{n} ||y_i - \sum_{l = 1}^{p}\sum_{j = 1}^{n}\theta_l K_l(x_j^{(l)},x_i^{(l)})u_j ||_{\cY}^2  + \\ \qquad  &\lambda_1\sum_{l = 1}^{p} \theta_l \sum_{i,j = 1}^{n} \langle K_l(x_i^{(l)},x_j^{(l)})u_i, u_j \rangle_{\cY} + \lambda_2 \sum_{l = 1}^{p} \theta_l,
    \end{aligned}
\end{equation}
where $\lambda_2 > 0$ is a sparsity-inducing regularization parameter (in the spirit of the Lasso). This formulation has two key benefits. First, it eliminates the need to pre-specify $\theta$, which can be challenging in high-dimensional settings. Second, by penalizing the magnitude of $\theta$, it naturally promotes sparsity and enables automatic variable selection. In particular, if the estimated weight $\hat{\theta}_l = 0$, the corresponding RKHS component can be removed from the regression function, and it implies that the corresponding covariate function $x_i^{(l)}$ for $i = 1,\dots,n$ has no contribution to the prediction of $y_i$. This mechanism allows the model to identify and retain only the most informative covariates in a data-driven manner.

\subsection{Block Coordinate Descent Optimization}
\label{sec:variable selection alg} 
To solve Eq. (\ref{eq:model 3}), we adopt a BCD algorithm. It iteratively solves $\bu$ with fixed $\btheta$ and solves $\btheta$ with fixed $\bu$, until convergence. The overall procedure is summarized in Algorithm \ref{alg:final BCD}, with implementation details provided in the following subsections.
\begin{algorithm}[h]
	\caption{Block coordinate descent for Eq. (\ref{eq:model 3})} 
	\label{alg:final BCD}
	\begin{algorithmic}
		\Require $\{(y_i, \bx_i)\}_{i=1}^n$, $g_l, l = 1,\ldots, p,T$, $\lambda_1, \lambda_2$, initial value $(\bu_0,\btheta_0)$, tolerance $\epsilon_{tol}$
	    \Ensure  $(\bu_k,\btheta_k)$
            \State Set $k = 0$
		\While{stopping criterion is not met} 
			\State $k = k + 1$
			\State Update $\bu_k = \mathop{\arg\min}\limits_{\bu\in\cY^{n}} m_{\btheta_{k-1}}(\bu)$ by solving Eq. (\ref{eq:solve u}) with Algorithm \ref{alg:solve u}.
			\State Update $\btheta_k = \mathop{\arg\min}\limits_{\btheta\in\bbR^{p}_{+}} h_{\bu_{k}}(\btheta)$ by solving Eq. (\ref{eq:solve theta}) with Algorithm \ref{alg:solve theta}.
			\If{$|q(\bu_{k},\btheta_{k}) - q(\bu_{k-1},\btheta_{k-1})|\leq \epsilon_{tol} |q(\bu_{k-1},\btheta_{k-1})|$}
                \State break
            \EndIf
		\EndWhile
	\end{algorithmic}
\end{algorithm}

\paragraph*{Solving $\bu$ with fixed $\btheta$.}
When $\btheta$ is fixed, the objective function in Eq.~\eqref{eq:model 3} reduces to: 
\begin{equation}
    \label{eq:solve u}
    \begin{aligned}
    \Tilde{\bu}  
   = &\mathop{\arg\min} \limits_{\bu \in \cY^n} \sum_{i = 1}^{n} ||y_i - \sum_{l = 1}^{p}\theta_l\sum_{j = 1}^{n} K_l(x_j^{(l)},x_i^{(l)})u_j ||_{\cY}^2 + \lambda_1\sum_{l = 1}^{p} \theta_l \sum_{i,j = 1}^{n} \langle K_l(x_i^{(l)},x_j^{(l)})u_i, u_j \rangle_{\cY} \\
       \equiv &\mathop{\arg\min}\limits_{\bu \in \cY^n} m_{\btheta}(\bu) 
    \end{aligned}
\end{equation}
This means that when $\btheta$ is given, we can also use Algorithm \ref{alg:solve u} to get the solution of Eq. (\ref{eq:solve u}).

\begin{algorithm} \caption{Spectral solver for Eq.~\eqref{eq:solve u}} \label{alg:solve u} \begin{algorithmic} \Require Training data $\{(y_i,\bx_i)\}_{i=1}^n$, scalar kernels $\{g_l\}_{l=1}^p$, common output operator $T$, weights $\btheta$, regularization parameter $\lambda_1$ \Ensure Solution $\hat{\bu}$ to Eq.~\eqref{eq:solve u sim} \State \textbf{Step 1.} Construct $\mathbf G_{\btheta}$ with entries $g(\bx_i,\bx_j)=\sum_{l=1}^p\theta_l g_l(x_i^{(l)},x_j^{(l)})$, and compute its eigenpairs $\{(\beta_b,\mathbf v_b)\}_{b=1}^n$. \State \textbf{Step 2.} Compute the nonzero eigenpairs $\{(\delta_q,w_q)\}_{q=1}^{\kappa}$ of $T$. \State \textbf{Step 3.} For each $b=1,\dots,n$ and $q=1,\dots,\kappa$, compute the projection coefficient $c_{bq}=\sum_{j=1}^n\langle v_{bj}w_q,y_j\rangle_{\cY}$, where $v_{bj}$ is the $j$th element of $\mathbf v_b$. \State \textbf{Step 4.} Reconstruct $\hat{\bu}=\sum_{b=1}^n\sum_{q=1}^{\kappa}(\beta_b\delta_q+\lambda_1)^{-1}c_{bq}(\mathbf v_b\otimes w_q)$. \end{algorithmic} \end{algorithm}

\paragraph*{Solving $\btheta$ with fixed $\bu$.} When $\bu$ is fixed, Eq.~\eqref{eq:model 3} becomes a convex quadratic program over the nonnegative orthant. To simplify notation, define
$A_{il}(\bu)=\sum_{j=1}^{n}K_l(x_j^{(l)},x_i^{(l)})u_j$ and
$B_l(\bu)=\sum_{i,j=1}^{n}\langle K_l(x_i^{(l)},x_j^{(l)})u_i,u_j\rangle_{\cY}$.
Then the $\btheta$-subproblem can be written as
\begin{equation}
\label{eq:solve theta}
\begin{aligned}
\tilde{\btheta}
= \arg\min_{\btheta\in\bbR_+^p}h_{\bu}(\btheta)
= \arg\min_{\btheta\in\bbR_+^p}
\left\{
\sum_{i=1}^{n}
\left\|y_i-\sum_{l=1}^{p}\theta_l A_{il}(\bu)\right\|_{\cY}^2
+\lambda_1\sum_{l=1}^{p}\theta_l B_l(\bu)
+\lambda_2\sum_{l=1}^{p}\theta_l
\right\}.
\end{aligned}
\end{equation}

The gradient of $h_{\bu}$ is
\begin{equation}
\label{eq:gradient of theta}
    \nabla h_{\bu}(\btheta)
    = \tilde{\mathbf D}_{\bu}
    +2\tilde{\mathbf K}_{\bu}\btheta
    +\lambda_2\mathbf 1,
\end{equation}
where $\tilde{\mathbf D}_{\bu}\in\bbR^p$ has $l$th element
$\tilde D_{\bu,l}=\lambda_1 B_l(\bu)-2\sum_{i=1}^{n}\langle y_i,A_{il}(\bu)\rangle_{\cY}$,
and $\tilde{\mathbf K}_{\bu}\in\bbR^{p\times p}$ has $(l,h)$th element
$\tilde K_{\bu,lh}=\sum_{i=1}^{n}\langle A_{il}(\bu),A_{ih}(\bu)\rangle_{\cY}$.
Here $\mathbf 1\in\bbR^p$ denotes the vector of ones. The Hessian of $h_{\bu}$ is $2\tilde{\mathbf K}_{\bu}$. Since for any $\mathbf a\in\bbR^p$,
$\mathbf a^T\tilde{\mathbf K}_{\bu}\mathbf a
=\sum_{i=1}^{n}\|\sum_{l=1}^{p}a_l A_{il}(\bu)\|_{\cY}^2\ge0$,
$\tilde{\mathbf K}_{\bu}$ is positive semidefinite. Hence Eq.~\eqref{eq:solve theta} is a convex optimization problem.

To solve the nonnegative quadratic program in Eq.~\eqref{eq:solve theta}, we use an active-set projected conjugate-gradient-type method, as detailed in Algorithm \ref{alg:solve theta}. It is an iterative algorithm, and we use $ \btheta_k $ to denote the value of $\btheta$ at the $k$-th iteration, with $\theta_{k,l}$ representing its $l$-th element. Correspondingly, $I_k = \{l|\theta_{k,l} = 0 \}$ and $I_k^c = \{1,\dots,p \} \backslash I_k$ denote the boundary points and interior points at the $k$-th iteration. For the conjugate gradient method, the most crucial aspect is computing the iterative direction, i.e. $\mathbf{d}_k$. For the boundary points, i.e. $l \in I_k$, we have 
\begin{equation}
\label{eq:alg 2 1}
    \begin{aligned}
        &d_{k,l}= 0, \quad \text{if}\  l \in I_k \ \text{and}\  \nabla_l h_{\bu}(\btheta_{k}) > 0; \quad  
        d_{k,l} = -\nabla_l h_{\bu}(\btheta_{k}),\quad   \text{if}\  l \in I_k \ \text{and}\  \nabla_l h_{\bu}(\btheta_{k}) \leq 0,
    \end{aligned}
\end{equation}
where $d_{k,l}$ is the $l$-th element of $\mathbf{d}_k$. 
$\nabla_l h_{\bu}(\btheta_{k})$ is the $l$-th element of $\nabla h_{\bu}(\btheta_{k})$. For the interior points, i.e. $l \in I_k^c$, we can use the modified Polak-Ribiere-Polyak method \citep{li2013conjugate} to obtain the following direction, 
\begin{equation}
\label{eq:alg 2 2}
\begin{aligned}
    \mathbf{d}_{k,I_k^c} &= -\nabla_{I_k^c}h_{\bu}(\btheta_{k}), \quad \text{if}\ k = 0;\quad 
    \mathbf{d}_{k,I_k^c} &= -\nabla_{I_k^c}h_{\bu}(\btheta_{k}) +\mathring{\beta}_k \mathbf{d}_{k-1,I_k^c}-\mathring{\theta}_k \bgamma_{k-1},  \quad  \text{if}\ k > 0,
\end{aligned}
\end{equation}
where  $\mathring{\beta}_k = \nabla_{I_k^c}h_{\bu}(\btheta_{k})^T \bgamma_{k-1}/||\nabla h_{\bu}(\btheta_{k-1}) ||^2 $, $\mathring{\theta}_k =\nabla_{I_k^c}h_{\bu}(\btheta_{k})^T\mathbf{d}_{k-1,I_k^c}/||\nabla h_{\bu}(\btheta_{k-1}) ||^2 $ and $\bgamma_{k-1} = \nabla_{I_k^c}h_{\bu}(\btheta_{k}) - \nabla_{I_k^c}h_{\bu}(\btheta_{k-1})$.

\begin{algorithm}[h]
    \caption{Active-set projected conjugate-gradient solver for Eq.~\eqref{eq:solve theta}}
    \label{alg:solve theta}
    \begin{algorithmic}
        \Require $\{(y_i, \bx_i)\}_{i=1}^n$, $g_l, l = 1,\ldots, p,T$, $\lambda_1, \lambda_2$, $\bu$, $\rho\in(0,1)$, initial value $\btheta_0$, tolerance $\epsilon_{tol}$
        \Ensure  Solution $\tilde{\btheta}$ to Eq. (\ref{eq:solve theta}) 
        \State Set $k = 0$
        \While{stopping criterion is not met}
        \State Calculate $\nabla h_{\bu}(\btheta_{k})$ by Eq. (\ref{eq:gradient of theta}) 
        \State Define $I_k=\{l:\theta_{k,l}=0\}$ and $I_k^c=\{1,\dots,p\}\setminus I_k$
        \State Compute $\mathbf{d}_k = (\mathbf{d}_{I_k},\mathbf{d}_{I_k^c} )$ by Eq. (\ref{eq:alg 2 1}) and Eq. (\ref{eq:alg 2 2})
        \State $\eta_k = 1$
        \While{$\btheta_k + \eta_k\mathbf{d}_k < 0 $ or  $h_{\bu}(\btheta_k + \eta_k\mathbf{d}_k) > h_{\bu}(\btheta_k) - \eta_k^2||\mathbf{d}_k||^2$}
        \State   $\eta_k = \eta_k \rho$ 
        \EndWhile
        \State  $\btheta_{k+1} = \btheta_{k} + \eta_k\mathbf{d}_k$
        \If{ $|h_{\bu}(\btheta_{k+1}) - h_{\bu}(\btheta_{k})|\leq \epsilon_{tol} |h_{\bu}(\btheta_{k})|$}
        \State break
        \EndIf
        \State $k = k + 1$
        \EndWhile
        \State $\Tilde{\btheta} =\btheta_k$ 
    \end{algorithmic}
\end{algorithm}

\begin{remark}[Computational complexity]
\label{rem:computation}
Let \(n\) be the number of training samples, \(p\) be the number of covariates, \(M_l\) be the grid size of the \(l\)-th covariate, and \(\kappa\) be the rank of the output operator \(T\). For each \(l=1,\ldots,p\), define the scalar kernel matrix \(\mathbf G_l\in\bbR^{n\times n}\) by \((\mathbf G_l)_{ij}=g_l(x_i^{(l)},x_j^{(l)})\). These scalar kernel matrices can be precomputed with cost \(O(n^2\sum_{l=1}^p M_l)\). Denote by \(C_T\) the cost of computing the leading \(\kappa\) eigenpairs of \(T\). In each outer BCD iteration, for fixed \(\btheta\), constructing \(\mathbf G_{\btheta}=\sum_{l=1}^p\theta_l\mathbf G_l\) and solving the Kronecker-structured spectral system costs \(O(pn^2+n^3+n^2\kappa)\), where the dominant term is usually the eigendecomposition of the \(n\times n\) matrix \(\mathbf G_{\btheta}\). For fixed \(\bu\), updating \(\btheta\) requires solving a \(p\)-dimensional nonnegative quadratic program. Using an active-set projected conjugate-gradient solver with \(N_{\mathrm{CG}}\) inner iterations, this step costs \(O(pn^2\kappa+p^2n\kappa+N_{\mathrm{CG}}p^2)\). Therefore, if the outer BCD algorithm takes \(N_{\mathrm{BCD}}\) iterations, the total training cost is \(O(n^2\sum_{l=1}^p M_l+C_T+N_{\mathrm{BCD}}[n^3+pn^2\kappa+p^2n\kappa+N_{\mathrm{CG}}p^2])\). When \(\kappa\) is small, the dominant per-iteration cost is the \(O(n^3)\) eigendecomposition of \(\mathbf G_{\btheta}\). 
\end{remark}
\section{Theoretical Properties}
\label{sec:theory}
Throughout the theoretical analysis, we use the following asymptotic notation. For two positive sequences $a_n$ and $b_n$, we write $a_n=o(b_n)$, or equivalently $b_n=\omega(a_n)$, if $\lim_{n\to\infty}a_n/b_n=0$. We write $a_n=O(b_n)$, or equivalently $b_n=\Omega(a_n)$, if $\limsup_{n\to\infty}|a_n|/|b_n|\le C$ for some constant $C>0$.

\subsection{Theoretical Properties of Eq.~\eqref{eq:model 3}}
We begin by introducing several assumptions on the true underlying parameters $(\bu^{*},\btheta^{*})$. Based on these assumptions, we analyze the statistical properties of the two block updates of the proposed algorithm by examining the updates of $\bu$ and $\btheta$ in Eq. \eqref{eq:solve u} and Eq. \eqref{eq:solve theta} separately. Specifically, we first study the convergence of $\bu$ while keeping $\btheta = \btheta^{*}$ fixed, and then analyze the convergence of $\btheta$ while keeping $\bu = \bu^{*}$ fixed. Since our primary interest lies in the variable selection aspect of the model, particular emphasis is placed on the convergence of the support of $\theta_l$, i.e., whether the algorithm correctly assigns zero to the components corresponding to irrelevant covariates.

After establishing block-wise convergence under fixed counterparts, we further investigate the convergence behavior of the support of $\btheta$ in the more realistic setting where $\bu$ is also iteratively updated, thereby characterizing the variable selection consistency of the full alternating procedure.

\begin{assumption}[Finite-sample well-specified RKHS representation]
    \label{ass:true value}
  Conditioning on the observed covariates $\{\bx_j\}_{j=1}^n$, there exists a nonempty equivalence class $S^{*} = \{(\bu,\btheta) |\bu \in \cY^n, \btheta \in \bbR^{p}_{+} \}$ consisting of all parameter pairs that satisfy the following property: Given any pair $(\tilde{y}, \tilde{\bx})$ of response and covariate functions, the conditional expectation of \( \tilde{y} \) is given by:
\begin{equation}
    \label{eq:true value 1}
              \bbE(\tilde{y}|\{\bx_{j} \}_{j=1}^n, \tilde{\bx}) = \sum_{l=1}^p\sum_{j=1}^n\theta_lK_l(x_j^{(l)},\tilde{x}^{(l)})u_j,
    \end{equation}   
where each kernel $K_l(\cdot,\cdot) = g_l(\cdot,\cdot) T$ is a known separable operator-valued kernel as defined in Eq.~\eqref{eq:separable kernel}.
Moreover, the set $S^{*}$ satisfies the following scale equivalence:
    \begin{equation}
    \label{eq:true value 2}
         \forall (\bu^{*},\btheta^{*}), (\bu_{*},\btheta_{*}) \in S^{*}, \exists c >0 \quad \text{s.t.} \quad  \bu^{*} = c \bu_{*}, \  \btheta^{*} = \btheta_{*}/c,
    \end{equation}
i.e., the model is invariant under scalar rescaling of $\bu$ and inverse rescaling of $\btheta$. 
\end{assumption}

In Assumption \ref{ass:true value}, the pair $(\tilde{y},\tilde{\bx})$ can represent either a new (unseen) test observation or one of the $n$ observed training samples. This assumption defines a set $S^{*}$ that contains all valid parameter pairs $(\bu, \btheta)$ consistent with the conditional expectation formulation in Eq.~\eqref{eq:true value 1}. It also reflects the scale nonidentifiability of the representation in Eq.~\eqref{eq:true value 1}. All pairs in $S^*$ yield the same fitted regression operator and the same support of $\btheta$. For example, if a covariate function $l$ is irrelevant for predicting the response, then $\theta_{l}^{*}=0$, and due to the scaling property, $\theta_l=0$ for all $(\bu,\btheta)\in S^{*}$. In summary, the structure of $S^*$ guarantees that zero components in $\btheta$ are invariant within the class, which is crucial for consistent variable selection. Therefore, in the following theoretical analysis, we may, without loss of generality, select an arbitrary representative pair from $ S^{*}$, denoted as $(\bu^{*},\btheta^{*})$, to serve as the true underlying parameter.

\begin{assumption} 
\label{ass:bounded_y}
The errors $\epsilon_i$ are independent mean-zero $Y$-valued
Gaussian random elements with covariance operator $\Gamma_\epsilon$
satisfying $\operatorname{tr}(\Gamma_\epsilon)<\infty$ and
$\|\Gamma_\epsilon\|_{\mathrm{op}}\le \sigma^2$. Moreover,
$\sup_i E\|y_i\|_\cY^2<\infty$.
\end{assumption}
\begin{assumption}[Regularization Parameter Rates]
    \label{ass:lambda rate}
    As $n \to \infty$, the regularization parameters satisfy:
    $\lambda_1 \to 0$, $\frac{\lambda_1}{n^{\frac{c_1-1}{2}}} \to \infty$, $\frac{\lambda_2}{n} \to 0 $, $\frac{\lambda_2}{n^{\frac{c_2+1}{2}}} \to \infty$, $\frac{n\lambda_1}{\lambda_2} \to 0$ for some constants $0\leq c_1, c_2 <1$.
\end{assumption}

Assumption~\ref{ass:lambda rate} governs the decay rates of the regularization parameters $\lambda_1$ and $\lambda_2$. 
The choice of $\lambda_2$ follows the convergence rate commonly used in the literature on Lasso-type penalized regression (e.g., \cite{zhao2006model}), ensuring that the second penalty term in Eq.~\eqref{eq:model 3} grows sublinearly with the sample size, i.e., $o(n)$. To maintain a balance between the two regularization terms, we set $\lambda_1$ to be $o(1)$, since the magnitude of $\sum_{l = 1}^{p} \theta_l \sum_{i,j = 1}^{n} \langle K_l(x_i^{(l)},x_j^{(l)})u_i, u_j \rangle_{\cY}$ is $O(n)$. The condition $\frac{n\lambda_1}{\lambda_2} \to 0$ ensures that the sparsity-inducing penalty dominates the bias introduced by the $\lambda_1$-regularization term in the support recovery analysis.

\begin{assumption}[Invertibility of the Scalar Kernel Matrix]
\label{ass:inverse operator}
We assume the operator $T$ has finite rank, and that there exist constants $b_{min}>0$ and $b_{max} > 0$ such that:
    \begin{equation}
        b_{min} \leq \sigma_{min}(\mathbf{G}_{\btheta^{*}}) \leq \sigma_{max}(\mathbf{G}_{\btheta^{*}}) \leq b_{max},
    \end{equation}
where $\sigma_{min}(\mathbf{G}_{\btheta^{*}})$ and $\sigma_{max}(\mathbf{G}_{\btheta^{*}})$ denote the minimum and maximum eigenvalues of $\mathbf{G}_{\btheta^{*}}$, respectively.
\end{assumption}
Assumption \ref{ass:inverse operator} ensures the invertibility of the kernel matrix $\mathbf{G}_{\btheta^{*}}$, which is necessary for analyzing convergence of the solution $\tilde{\bu}$. 

\begin{theorem}[Average coefficient stability of $\tilde{\bu}$ with Fixed $\btheta^{*}$]
\label{thm:u}
Suppose we fix $\btheta = \btheta^{*}$ and solve Eq. (\ref{eq:solve u}) to obtain the solution $\Tilde{\bu}$. Under Assumptions \ref{ass:true value}-\ref{ass:inverse operator}, we have,
    \begin{equation}
        \bbE\left[||\frac{1}{n}\sum_{i=1}^n \Tilde{u}_i - \frac{1}{n}\sum_{i=1}^n u_i^{*}||_{\cY}^2 \right] = O(\lambda_1^2)\to 0,\ \text{as}\ n \to \infty,
    \end{equation}
where $\Tilde{u}_i$ is the $i$-th element of $\Tilde{\bu}$, and $u_i^{*}$ is the corresponding element in $\bu^{*}$.
\end{theorem}
 \begin{proof}{Proof}
     The proof of Theorem \ref{thm:u} is shown in Appendix 3. \Halmos
 \end{proof}

Theorem~\ref{thm:u} establishes convergence of the averaged coefficient element when the true kernel-weight vector $\btheta^*$ is fixed. It should be interpreted as a weak coefficient-level stability result for the fixed-weight subproblem.

\begin{remark} The model for estimating $\tilde{\bu}$ with fixed $\btheta=\btheta^{*}$ in Eq.~\eqref{eq:solve u} is structurally identical to the model in Eq.~\eqref{eq:solve u sim}, where $\btheta$ is treated as a fixed hyperparameter. Therefore, Theorem~\ref{thm:u} also implies convergence of the solution $\hat{\bu}$ in Eq.~\eqref{eq:solve u sim}, provided that the true hyperparameter $\btheta = \btheta^*$ is correctly chosen. 
\end{remark}

To analyze the convergence properties of the estimated kernel weights $\btheta$, we introduce the sign function, denoted $sign(\cdot)$, which maps each positive element to 1 and zero element to 0. Without loss of generality, we assume that the first $r$ elements of the true parameter vector $\btheta^{*} = (\theta_1^{*},\dots,\theta_p^{*})^T$ are strictly positive, while the remaining $p-r$ elements are zero. This implies that the last $p-r$ covariate functions are irrelevant to $y$ and should be excluded during the variable selection process. Under this setting, the support of $\btheta^{*}$ is represented by $sign(\btheta^{*}) = (\underbrace{1,\dots,1}_{r},\underbrace{0,\dots,0}_{p-r})^T$.

\begin{assumption}[Irrepresentable Condition] 
\label{ass:irrepresentable}
For $\bu^*$, define the normalized matrix $\bC_{\bu^*}^n = \frac{\Tilde{\mathbf{K}}_{\bu^*}}{n}$ where $\tilde{\mathbf{K}}_{\bu^*}$ is defined as in Eq.~\eqref{eq:gradient of theta}. We assume the matrix $\bC_{\bu^*}^n$ converges to a positive definite matrix
$\bC_{\bu^*}$.
Furthermore, we partition $\bC_{\bu^*}$ according to the active (relevant) and inactive (irrelevant) sets of covariates, i.e., 
$\begin{pmatrix}
    \mathbf{C}_{\bu^*,11} & \mathbf{C}_{\bu^*,12}\\
    \mathbf{C}_{\bu^*,21}& \mathbf{C}_{\bu^*,22}
\end{pmatrix}$
where $\mathbf{C}_{\bu^*,11} \in \bbR^{r\times r}$ and $\mathbf{C}_{\bu^*,21}  \in \bbR^{(p-r)\times r}$ correspond to the relevant and irrelevant components, respectively. We assume there exists a constant $\xi_{\bu^*} \in (0,1)$ such that the following condition holds:
    \begin{equation}
    \label{eq:ass irre}
        ||\mathbf{C}_{\bu^*,21}(\mathbf{C}_{\bu^*,11})^{-1}sign(\btheta_{(1)}^{*}) ||_{\infty} \leq 1 - \xi_{\bu^*},
    \end{equation}
where $sign(\btheta_{(1)}^{*})\in \bbR^{r}$ denotes the subvector of the first $r$ elements of $sign(\btheta^{*})$, and $||\cdot||_{\infty}$ is the infinity norm.
\end{assumption}

Assumption \ref{ass:irrepresentable} can be viewed as the kernel-based analogue of the irrepresentable condition introduced in \cite{zhao2006model} for linear Lasso models. 
In the standard Lasso setting, this condition serves as a crucial constraint that ensures irrelevant covariates cannot be linearly expressed by the relevant ones to a large extent. This helps prevent excessive leakage of signal from the relevant set into the irrelevant set, which would otherwise impede correct variable selection. In our functional setting, 
the term $\mathbf{C}_{\bu^*,21}(\mathbf{C}_{\bu^*,11})^{-1}$ can be interpreted approximately as the regression coefficients obtained by projecting the irrelevant covariate functions onto the span of the relevant ones using kernel-based representations. Therefore, the inequality in Eq.~\eqref{eq:ass irre} enforces that the total representational contribution of any irrelevant covariate function (by the relevant ones) remains strictly less than 1. This is a reasonable condition, since if the irrelevant covariates are highly correlated (in the RKHS sense) with the relevant ones, it becomes fundamentally difficult to distinguish and select the truly informative covariates. It is consistent with the insights in \cite{zhao2006model}.

\begin{theorem}[Support Recovery with Fixed $\bu^{*}$]
\label{thm:theta}
    Fix $\bu = \bu^{*}$. Under Assumptions \ref{ass:true value}-\ref{ass:lambda rate}, \ref{ass:irrepresentable}, the solution $\tilde{\btheta}$ to Eq.~\eqref{eq:solve theta} satisfies
    \begin{equation}
        P(sign(\Tilde{\btheta}) = sign(\btheta^{*})) = 1 - o(e^{-n^{c_2}}) \to 1, as\  n \to \infty,
    \end{equation}
    where $c_2$ is the constant specified in Assumption \ref{ass:lambda rate}.
\end{theorem}
\begin{proof}{Proof}
The proof of Theorem \ref{thm:theta} is shown in Appendix 4. \Halmos
\end{proof}
Theorem~\ref{thm:theta} establishes that when the true coefficient functions $\bu^{*}$ are known, the estimated weights $\tilde{\btheta}$ will recover the correct sparsity pattern, i.e., identify the relevant covariate functions, with high probability as $n \to \infty$. In particular, the probability of exact support recovery approaches 1 at an exponential rate.

Theorem \ref{thm:theta} isolates the variable-selection mechanism by analyzing the
oracle problem with $u^*$ fixed. The remaining question is whether this
oracle support recovery result can be transferred to the actual joint
estimator $(\hat{\bu},\hat{\btheta})$. Theorem \ref{thm:theta global} answers this question
affirmatively under a local consistency condition on $\hat \bu$.

\begin{theorem}[Support Recovery of the Joint Estimator under Local Consistency of $\hat \bu$]
\label{thm:theta global}
Suppose Assumptions \ref{ass:true value}-\ref{ass:irrepresentable} hold.
Assume further that
$\|\hat \bu-\bu^\ast\|_n=o_p(1)$ where $\|\bu\|_n^2:=\frac1n\sum_{i=1}^n\|u_i\|_{\cY}^2$.
Then the global minimizer $(\hat \bu,\hat{\btheta})$ of Eq.~(11) satisfies
\[
\Pr\!\bigl(sign(\hat{\btheta})=sign(\btheta^\ast)\bigr)\to 1, \quad \text{as}\quad n\to \infty.
\]
\end{theorem}
\begin{proof}{Proof}
The proof of Theorem \ref{thm:theta global} is shown in Appendix 5. \Halmos
\end{proof}

Theorem~\ref{thm:theta global} is a local support recovery result. The additional condition $\|\hat\bu-\bu^*\|_n=o_p(1)$ is imposed to isolate the effect of coefficient-estimation error on sparse kernel-weight recovery. Since establishing global norm consistency for the nonconvex joint estimator would require stronger global identifiability and curvature conditions, we leave this extension for future work.

\subsection{Theoretical Properties of Algorithm \ref{alg:final BCD}}
The preceding theoretical results guarantee desirable properties of the optimal solution to Eq. (\ref{eq:model 3}). However, because the objective function $q(\bu,\btheta)$ is not jointly convex in both variables,  it is generally difficult to ensure that Algorithm~\ref{alg:final BCD} will converge to the global optimum in practice.
To facilitate a convergence analysis of Algorithm~\ref{alg:final BCD}, we first establish the following property of partial convexity.

\begin{proposition}[Partial Convexity]
\label{prop:convex}
    For the objective function $q(\bu,\btheta)$ in Eq. (\ref{eq:model 3}), the following hold:
    
     (i) $m_{\btheta}(\bu)$ is convex in $\bu$ for a fixed $\btheta$;  
     (ii) $h_{\bu}(\btheta)$ is convex in $\btheta$ for a fixed $\bu$.
\end{proposition}
\begin{proof}{Proof}
    The proof of Proposition \ref{prop:convex} is shown in Appendix 6. \Halmos
\end{proof}
Proposition~\ref{prop:convex} establishes that while the joint objective is not convex, the subproblems encountered during block coordinate descent are convex. However, to ensure stronger convergence guarantees for Algorithm~\ref{alg:final BCD}, we impose an additional assumption of uniform strong convexity on the subproblems.

\begin{assumption}[Uniform strong convexity and exact block updates]
\label{ass:strongly convex}
Let $\{(\bu_k,\btheta_k)\}_{k\ge0}$ be the sequence generated by Algorithm~\ref{alg:final BCD}. Assume that $m_{\btheta_k}(\bu)$ is strongly convex in $\bu$ with modulus $M_k$ and $h_{\bu_k}(\btheta)$ is strongly convex in $\btheta$ with modulus $H_k$, where $0<\underline M\le M_k\le\overline M<\infty$ and $0<\underline H\le H_k\le\overline H<\infty$. We also assume that each block subproblem is solved exactly, or that the accumulated inner optimization errors are summable.
\end{assumption}

Assumption \ref{ass:strongly convex} ensures that the subproblems retain strict curvature across iterations, preventing degeneracy into merely convex problems. The lower bounds $\underline{M}$ and $\underline{H}$ ensure that the strong convexity strictly holds, while the upper bounds $\overline{M}$ and $\overline{H}$ ensure that the curvature does not grow uncontrollably. These bounds are also helpful when initializing $(\bu_0,\btheta_0)$ within a well-behaved region.

\begin{theorem}[Convergence to a Stationary Point]
\label{thm:BCD}
Under Assumption~\ref{ass:strongly convex}, every limit point of the sequence $\{(\bu_k, \btheta_k)\}_{k \geq 0}$ generated by Algorithm~\ref{alg:final BCD} is a stationary point of Eq.~\eqref{eq:model 3}. 
\end{theorem}
\begin{proof}{Proof}
The proof of Theorem \ref{thm:BCD} is shown in Appendix 7.
\end{proof}
Theorem~\ref{thm:BCD} shows that every limit point of the BCD sequence is stationary. If the sequence has a unique limit point, then the whole sequence converges to that stationary point.

\section{SIMULATION}
\label{sec:simulation}
We conduct simulation experiments under three scenarios to evaluate the performance of FRCOMA:
\begin{itemize}
\item \textbf{One-way}: We set $p=5$, and all covariates and the response variable are continuous one-way data (i.e., $d=1$). 
\item \textbf{Multi-way-small}: We set $p=5$, and both the covariates and the response are continuous multi-way data. 
\item \textbf{Multi-way-large}: We increase the number of covariates to $p=16$, with all covariates and the response still being continuous multi-way data. 
\end{itemize}

We compare FRCOMA with representative methods from the main methodological families reviewed in Section~\ref{sec:literature}. The baselines are grouped according to whether they model functional structure, tensor structure, or nonlinear kernel structure.

First, we include functional regression baselines that are designed for continuous one-way data: \textbf{MFPCA} \citep{chiou2014multivariate} uses multivariate functional principal component analysis to represent functional predictors and responses in a reduced-dimensional score space, followed by regression on the scores. \textbf{FoF-SC} \citep{hong2011inference} is a sparse function-on-function regression method that imposes a Lasso penalty but assumes scalar regression coefficients for covariates. \textbf{FoF} \citep{cai2021variable} estimates functional regression coefficients through FPCA and performs variable selection using a group SCAD penalty. These methods preserve functional continuity but are primarily designed for continuous one-way domains.

Second, we include tensor regression baselines that are designed for discrete multi-way arrays: \textbf{MToT} \citep{gahrooei2021multiple} models the relationship between tensor-valued covariates and responses using low-rank Tucker representations. \textbf{MToT-VC} \citep{gahrooei2021multiple} extends MToT with a thresholding-based covariate selection rule using the Frobenius norm of estimated coefficient tensors. These methods directly handle multi-way structure but rely on fixed-grid tensor representations and do not model the underlying continuous domains.

Third, we include a nonlinear tensor-kernel baseline: \textbf{KMTPLS} \citep{hou2016common} uses kernel partial least squares with multi-mode product kernels to capture nonlinear relationships between tensor-valued predictors and responses. It provides a nonlinear comparison method but does not support continuous operator-valued functional output modeling or covariate-level variable selection.

Note that MFPCA, FoF-SC, and FoF are designed for continuous one-way data only. For multi-way scenarios, we reshape multi-way data into one-way representations before applying them. For MToT, MToT-VC, and KMTPLS, which can only handle discrete multi-way data, we discretize each mode of continuous data using 100 evenly spaced points to get their discrete counterpart.

The detailed simulation results are reported below. Due to space limitations, the running-time comparison and its detailed discussion are deferred to Appendix~8.

\subsection{Continuous One-way Scenario}
\label{sec:case1}
We consider all $y_i$ and $x_i^{(l)}$ to be continuous one-way data, i.e. $d_y=1,d_l =1,l = 1,\dots,p$ with $p=5$. We use the following model to generate data:
\begin{equation}
\label{eq:generate data}
    y_i = \sum_{l = 1}^p\sum_{j = 1}^n \theta_l g_l(x_j^{(l)},x_i^{(l)})Tu_j + \epsilon_i, i = 1,\dots,n, 
\end{equation}
where $x_{i}^{(l)}, u_j$ and $\epsilon_i$ are defined as follows. 
Each covariate function $\{x_{i}^{(l)}\}_{i=1}^n$ is generated from a Gaussian process on $[0,1]$ with zero mean and covariance function $cov(x^{(l)}(s),x^{(l)}(t)) = \gamma_lexp(- \alpha_l(s-t)^2)$ where $\gamma_l$ and $\alpha_l$ are generated randomly from a uniform distribution $U(1,2)$ for $l = 1,\dots,p$. The kernel functions $g_l(\cdot,\cdot)$ are chosen from three options: (1) Gaussian: $g_l(x_i^{(l)},x_j^{(l)}) = \exp(-\frac{||x_i^{(l)}-x_j^{(l)} ||^2_{\cX_l}}{\sigma^2_g})$; (2) Cauchy: $g_l(x_i^{(l)},x_j^{(l)}) = \frac{1}{1+\frac{||x_i^{(l)}-x_j^{(l)} ||^2_{\cX_l}}{\sigma^2_g}}$; and (3) Exponential: $g_l(x_i^{(l)},x_j^{(l)}) = \exp(-\frac{||x_i^{(l)}-x_j^{(l)} ||_{\cX_l}}{\sigma^2_g})$, where $||\cdot ||_{\cX_l}$ denotes the $L^2$ norm on $\cX_l$ and $\sigma_g=1$.  For each experiment, all $g_l, l = 1,\dots,p$ use the same kernel.  The operator $T$ is chosen as a projection operator onto a sine basis: $Tu(t) = \sum_{h=1}^{50}\langle b_h,u \rangle_{\cY} b_h(t)$ where $b_h(t) = sin(2\pi h t)$. 
$u_j$ is defined as $u_j(t) = sin(2\pi j t),j = 1,\dots,n, t\in [0,1]$. The noise $\epsilon_i \sim N(0,\sigma^2)$ with $\sigma^2 \in \{0.5, 1\}$. We define three sparsity settings for $\btheta$: $M = \{1,2,3,4\}, \{2,4\}$ and $\{2\}$, where $M =\{l|\theta_l \neq 0\} $. Nonzero $\theta_l$ values are drawn from $U(1,2)$. 

For each setting, we generate $n = 50$ samples. FRCOMA and baseline methods are applied to perform both prediction and variable selection. FRCOMA’s hyperparameters $\lambda_1$ and $\lambda_2$ are selected via 5-fold cross-validation, typically falling within the ranges $\lambda_1\in [0.1,0.3]$ and $\lambda_2 \in [1,3]$, consistent with Assumption~\ref{ass:lambda rate}.

For each parameter setting and each baseline, we replicate the experiment 100 times and calculate the F1-score for variable selection as well as the mean and standard deviation of the MSE. MSE is computed as $MSE = \frac{1}{n}\sum_{i = 1}^n ||\hat{f}(\bx_i) - y_i ||^2_\cY$. The results are presented in Figure \ref{fig:one-way}. 
\begin{figure}
\centering{
\includegraphics[width=\linewidth]{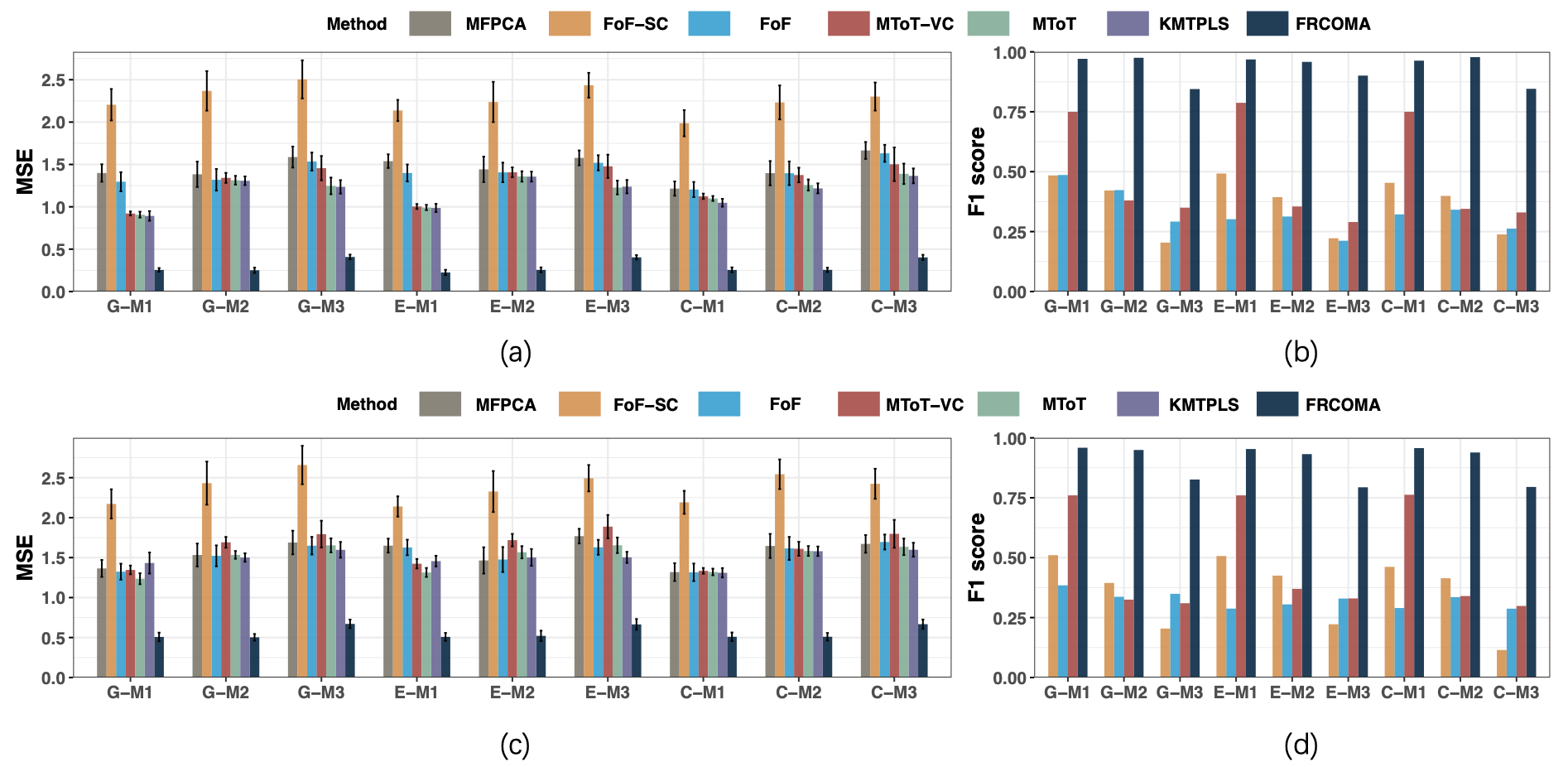}}
    \caption{\label{fig:one-way}One-way scenario results: (a) MSE under $\sigma^2 = 0.5$; (b) F1-score under $\sigma^2 = 0.5$; (c) MSE under $\sigma^2 = 1$; (d) F1-score under $\sigma^2 = 1$. G, E, and C indicate Gaussian, Exponential, and Cauchy kernels. M1, M2, and M3 correspond to \(M=\{1-4\}\), \(\{2, 4\}\), and \(\{2\}\), respectively.}    
\end{figure}

The simulation results clearly demonstrate that FRCOMA outperforms all baseline methods in both variable selection accuracy and prediction error. This advantage stems from its ability to model nonlinear relationships while preserving the continuity of the data. Overall, as the noise variance $\sigma^2$ increases, all methods experience higher MSE, and as the number of true relevant covariates decreases, F1-scores also drop. These trends are consistent with theoretical expectations.

For baseline comparison, among the function-on-function regression baselines, both FoF and MFPCA perform poorly due to their linear modeling assumptions, which limit their ability to capture complex dependencies. FoF-SC shows the worst performance, as its scalar coefficient assumption fails to capture data intra-mode correlations, severely impacting both selection and fitting. For tensor-on-tensor methods, MToT-VC underperforms relative to MToT, particularly when fewer covariates are truly relevant. This reflects the poor covariate selection ability of MToT-VC. As for KMTPLS, while it captures certain nonlinear relationships, it still performs worse than FRCOMA. This is mainly because it discretizes continuous data, ignoring their inherent smoothness and continuity. 

\subsection{Continuous Multi-way-small Scenario}
\label{sec:simulation 2}
In this setting, the response $y_i$ is a continuous two-way array and the covariates $x_i^{(l)},l = 1,\dots,5$ are continuous multi-way arrays with varying dimensions: $d_1 = 1, d_2 = 2, d_3 = 3, d_4 = 1,d_5 = 2$.  We also use Eq. (\ref{eq:generate data}) to generate the response $y_i$. Each $x_{i}^{(l)}$ is generated from a Gaussian process on $[0,1]^{d_l}$ with zero mean and covariance: $cov(x^{(l)}(\mathbf{s}),x^{(l)}(\mathbf{t})) = \gamma_l\exp(- \alpha_l||\mathbf{s}-\mathbf{t} ||_2^2), \gamma_l, \alpha_l \sim U(1,2)$. 
We use the same settings for $g_l$, $\btheta$, $M$, and $\epsilon$, as in Section \ref{sec:case1}. The operator $T$ is defined as $Tu(s,t) = \sum_{h = 1}^7\sum_{l=1}^7 \langle b_{hl},u\rangle_{\cY}b_{hl}(s,t)$, where $b_{hl}(s,t)=sin(2\pi hs)sin(2\pi lt)$. We define $u_j(s,t) = sin(2\pi hs)sin(2\pi lt),j = 1,\dots,n,(s,t)\in [0,1]^2$ with $h = \lceil j/7\rceil,l = j - 7(\lceil j/7\rceil-1)$, where $\lceil\cdot\rceil$ indicates the ceiling operator.

For each configuration, we generate $n = 50$ samples. The hyperparameters of FRCOMA are tuned via 5-fold cross-validation, with $\lambda_1\in [0.1,0.4]$ and $\lambda_2\in [1,2]$. Each experiment is repeated 100 times. We report the average F1-score and MSE (mean ± std), as shown in Figure~\ref{fig:multi-way-small}.

\begin{figure}
\centering{\includegraphics[width=\linewidth]{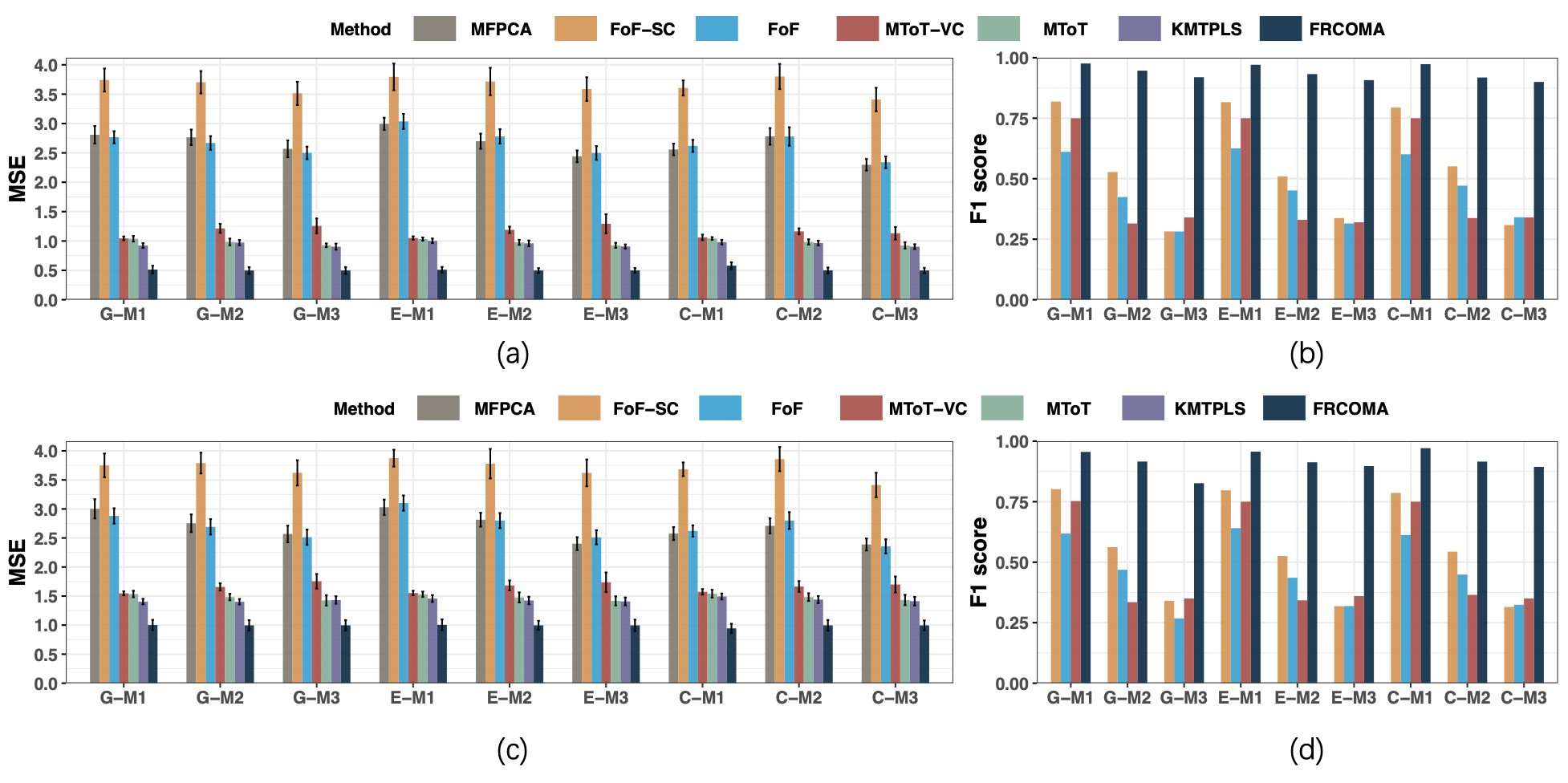}}
    \caption{\label{fig:multi-way-small}Multi-way-small scenario results: (a) MSE under $\sigma^2 = 0.5$; (b) F1-score under $\sigma^2 = 0.5$; (c) MSE under $\sigma^2 = 1$; (d) F1-score under $\sigma^2 = 1$. G, E, and C refer to the Gaussian, Exponential, and Cauchy kernels. M1, M2, and M3 represent the configurations of \(M\): \(\{1 - 4\}\), \(\{2, 4\}\), and \(\{2\}\), respectively.}
\end{figure}

Compared to the one-way case, FRCOMA continues to show strong performance in both prediction and variable selection with multi-way data. However, function-on-function baselines (MFPCA, FoF, FoF-SC) suffer significant performance drops. This is expected, as these methods flatten multi-way inputs into one-way functions, losing key structural information. Additionally, we observe a sharper decline in F1-scores as the number of relevant covariates decreases, highlighting the increased difficulty of variable selection in multi-way settings.

\subsection{Continuous Multi-way-large Scenario}
This setting extends Section~\ref{sec:simulation 2} by increasing the number of covariates to $p=16$. The response $y_i$ is a continuous two-way array, and each covariate $x_{i}^{(l)}$ is a continuous multi-way array with dimension $d_l$ which is uniformly randomly sampled from $\{1, 2, 3\}$. We use the same data generation model (Eq.~\ref{eq:generate data}), with each $x_i^{(l)}$ sampled from a Gaussian process on $[0,1]^{d_l}$, using zero mean and covariance $cov(x^{(l)}(\mathbf{s}),x^{(l)}(\mathbf{t})) = \gamma_l \exp(- \alpha_l||\mathbf{s}-\mathbf{t} ||_2^2)$ where $\gamma_l, \alpha_l \sim U(1,2)$.
The kernels $g_l$, noise $\epsilon$ and operator $T$ follow the same structure as earlier settings. Specifically, $T$ projects onto 2D sine bases: $Tu(s,t) = \sum_{h = 1}^8\sum_{l=1}^8 \langle b_{hl},u\rangle_{\cY}b_{hl}(s,t)$ where $b_{hl}(s,t)=sin(2\pi hs)sin(2\pi lt)$. We define $u_j(s,t) = sin(2\pi hs)sin(2\pi lt),j = 1,\dots,n,(s,t)\in [0,1]^2$ with $h = \lceil j/8\rceil,l = j - 8(\lceil j/8\rceil-1)$. 
We test three configurations of relevant covariates: $M = \{1-8\}, \{1-4\},\{1,2\}$ with nonzero $\theta_l\sim U(1, 2)$. 

For each configuration, we generate $n=70$ samples. The hyperparameters of FRCOMA are selected via 5-fold cross-validation, with $\lambda_1 \in [0.1, 0.3]$ and $\lambda_2 \in [2,3]$. Each experiment is repeated 100 times, and we report average MSE and F1-score across replicates. Results are shown in Figure~\ref{fig:multi-way-big}.

\begin{figure}
    \centering{\includegraphics[width=\linewidth]{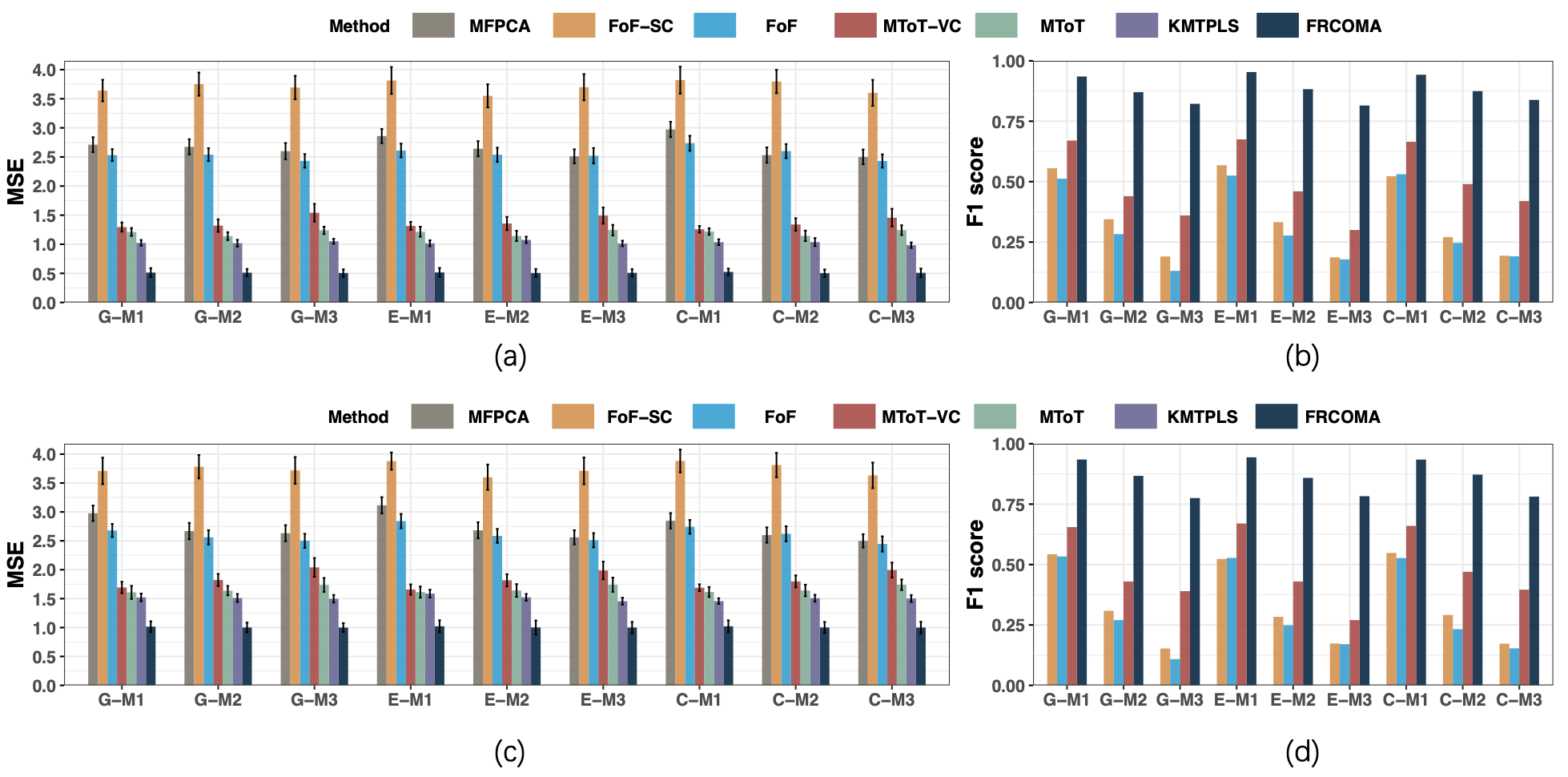}}
    \caption{\label{fig:multi-way-big}Multi-way-large scenario results: (a) MSE under $\sigma^2 = 0.5$; (b) F1-score under $\sigma^2 = 0.5$; (c) MSE under $\sigma^2 = 1$; (d) F1-score under $\sigma^2 = 1$. G, E, and C indicate Gaussian, Exponential, and Cauchy kernels. M1, M2, and M3 correspond to $M= \{1-8\}, \{1-4\}, \{1,2\}$, respectively.}
\end{figure}

FRCOMA maintains strong performance in both prediction and variable selection. However, as the number of covariates increases, selection becomes more difficult, and F1-scores decline across all methods. This effect is more pronounced when fewer covariates are truly relevant, reflecting the increased complexity of high-dimensional functional data. Notably, KMTPLS shows improved relative performance in this setting, as its ability to model nonlinear relationships becomes more advantageous with more covariates.

\section{CASE STUDY}
\label{sec:case}
\subsection{Movie Trailer Dataset}
\label{sec:video}
We evaluate the performance of FRCOMA on a real-world dataset derived from movie trailers on YouTube. Both the response and covariates are continuous one-way data, with $d_y = d_1 = \cdots = d_p = 1$. The goal is to predict the engagement curve (``heat curve'') of a trailer—representing user engagement and rewatch likelihood—based on multimodal features extracted from the video and audio content. Accurate prediction of the heat curve can guide trailer editing to highlight attention-grabbing segments before release. Given the complexity and redundancy in the data, both nonlinear modeling and variable selection are essential. 

We use 19 continuous one-way array covariates ($p=19$): 1) visual features: grayscale pixel correlation ($x^{(1)}$), color histogram distance ($x^{(2)}$), ResNet-based visual similarity ($x^{(3)}$), colorfulness ($x^{(10)}$), hue value related metric ($x^{(12)}$), saturation ($x^{(13)}$), proportion of black color ($x^{(14)}$), brightness ($x^{(15)}$ ), contrast ($x^{(16)}$)  and several edge-based metrics ($x^{(11)},x^{(17)}, x^{(18)}, x^{(19)}$); 2) audio features:  total sound energy ($x^{(4)}$), pitch ($x^{(5)}$), vocal-specific energy and pitch ($x^{(6)},x^{(7)}$), and their temporal changes ($x^{(8)},x^{(9)}$).
More detailed variable explanations are provided in Appendix 9.

We categorize 450 movie trailers into three time periods: 2000–2009, 2010–2019, and post-2020, with 150 trailers in each. For each group, we randomly split the data into 100 training and 50 testing samples. For FRCOMA, we use a Gaussian kernel, $g_l(x_i^{(l)},x_j^{(l)}) = \exp(-\frac{||x_{i}^{(l)}-x_j^{(l)} ||^2_{\cX_l}}{\sigma_g^2}), l = 1,\dots,p$, and define the operator $T$ as a projection onto a trigonometric basis, i.e., $Tu(t) = \sum_{h=0}^{50}\langle b_h,u \rangle_{\cY} b_h(t)$ where $\{b_{h}(t)=cos(2\pi ht)|h = 0,\dots,25\}\cup \{b_{h}(t)=sin(2\pi (h-25)t)|h = 26,\dots,50\} $. The hyperparameters $\lambda_1,\lambda_2, \sigma_g$ are selected via 5-fold cross-validation. For baselines, we adjust their penalty or number of components to match FRCOMA's training MSE, then compare test MSEs. Since the true number of relevant covariates is unknown, the MToT-VC method is not applicable. The results are shown in Table \ref{tab:case video}.

\begin{table*}[hbt]
    \centering
    \caption{Variable selection and fitting results of Movie Trailer Dataset}
    \resizebox{\linewidth}{!}{
    \begin{tabular}{lllllllllllllllllllllllll}
    \hline
    Year & Method & $x^{(1)}$ &  $x^{(2)}$ & $x^{(3)}$ & $x^{(4)}$ & $x^{(5)}$& $x^{(6)}$&  $x^{(7)}$ &  $x^{(8)}$&  $x^{(9)} $&  $x^{(10)}$ &  $x^{(11)}$&  $x^{(12)}$ &  $x^{(13)}$ &  $x^{(14)}$ &  $x^{(15)}$ &  $x^{(16)}$ &  $x^{(17)}$ &  $x^{(18)}$ &  $x^{(19)}$ & $MSE_{train}$ & $MSE_{test}$\\
    \hline
    \multirow{6}{*}{$\leq2009$}& MFPCA &/ &/ &/ &/ &/ &/ &/ & /& /& / & /&/ &/ &/ &/ &/ &/ &/ &/ &0.0021 &0.0387\\
    ~ & FoF-SC & \ding{51} &\ding{51}&\ding{51}&\ding{51}&\ding{51}&\ding{51}&\ding{51}&\ding{51}&\ding{51}&\ding{51} & \ding{51} &\ding{51}&\ding{55}&\ding{51}&\ding{51}&\ding{51}&\ding{51}&\ding{51}&\ding{55}&0.0049 &0.0671\\
    ~ & FoF & \ding{51} &\ding{55} &\ding{51}&\ding{51}&\ding{51}&\ding{51}&\ding{51}&\ding{51}&\ding{51}&\ding{51} & \ding{51} &\ding{51}&\ding{51}&\ding{51}&\ding{51}&\ding{51}&\ding{55}&\ding{51}&\ding{51}&0.0018 &0.0377 \\
    ~& MToT &/ &/ &/ &/ &/ &/ &/ & /& /& / & /&/ &/ &/ &/ &/ &/ &/ &/ &0.0015 &0.0408 \\
    ~& KMTPLS &/ &/ &/ &/ &/ &/ &/ & /& /& / & /&/ &/ &/ &/ &/ &/ &/ &/ &0.0016 &0.0374\\
    ~& FRCOMA &\ding{51} &\ding{51} &\ding{55} &\ding{51} &\ding{51} &\ding{51} &\ding{51} & \ding{51}& \ding{51}& \ding{51} &\ding{51} &\ding{51} &\ding{55} &\ding{51} &\ding{51} &\ding{55} &\ding{51} & \ding{51}& \ding{55}& 0.0015& \textbf{0.0240} \\
    \hline
    \multirow{6}{*}{$2010-2019$}& MFPCA &/ &/ &/ &/ &/ &/ &/ & /& /& / & /&/ &/ &/ &/ &/ &/ &/ &/ &0.0019 &0.0351\\
    ~ & FoF-SC & \ding{51} &\ding{51}&\ding{55}&\ding{51}&\ding{51}&\ding{51}&\ding{51}&\ding{51}&\ding{51}&\ding{51}  & \ding{51} &\ding{51}&\ding{51}&\ding{51}&\ding{51}&\ding{51}&\ding{51}&\ding{51}&\ding{55}&0.0063 &0.0735\\
    ~ & FoF & \ding{51} &\ding{51} &\ding{51}&\ding{51}&\ding{51}&\ding{51}&\ding{51}&\ding{51}&\ding{51}&\ding{51}& \ding{51} &\ding{51}&\ding{51}&\ding{51}&\ding{51}&\ding{51}&\ding{51}&\ding{55}&\ding{51}&0.0015 &0.0329  \\
    ~& MToT &/ &/ &/ &/ &/ &/ &/ & /& /& / & /&/ &/ &/ &/ &/ &/ &/ &/ &0.0021 &0.0356\\
    ~& KMTPLS &/ &/ &/ &/ &/ &/ &/ & /& /& / & /&/ &/ &/ &/ &/ &/ &/ &/ &0.0016 &0.0326\\
    ~& FRCOMA &\ding{51} &\ding{51} &\ding{55} &\ding{51} &\ding{51} &\ding{51} &\ding{51} & \ding{51}& \ding{51}& \ding{51} &\ding{51} &\ding{51} &\ding{55} &\ding{51} &\ding{51} &\ding{55} &\ding{51} & \ding{51}& \ding{55}& 0.0015& \textbf{0.0240}\\
    \hline
    \multirow{6}{*}{$\geq 2020$}& MFPCA &/ &/ &/ &/ &/ &/ &/ & /& /& / & /&/ &/ &/ &/ &/ &/ &/ &/ &0.0036 &0.0297\\
    ~ & FoF-SC & \ding{51} &\ding{51}&\ding{55}&\ding{51}&\ding{51}&\ding{51}&\ding{51}&\ding{51}&\ding{51}&\ding{51}  & \ding{51} &\ding{51}&\ding{55}&\ding{51}&\ding{51}&\ding{51}&\ding{51}&\ding{51}&\ding{51}&0.0058 &0.0689 \\
    ~ & FoF& \ding{51} &\ding{51} &\ding{51}&\ding{51}&\ding{51}&\ding{51}&\ding{51}&\ding{51}&\ding{51}&\ding{51} & \ding{51} &\ding{51}&\ding{51}&\ding{51}&\ding{51}&\ding{51}&\ding{51}&\ding{55}&\ding{51}&0.0041 &0.0305\\
    ~& MToT &/ &/ &/ &/ &/ &/ &/ & /& /& / & /&/ &/ &/ &/ &/ &/ &/ &/ &0.0039 &0.0386 \\
    ~& KMTPLS &/ &/ &/ &/ &/ &/ &/ & /& /& / & /&/ &/ &/ &/ &/ &/ &/ &/ &0.0035 &0.0294 \\
    ~& FRCOMA &\ding{51} &\ding{51} &\ding{55} &\ding{51} &\ding{51} &\ding{51} &\ding{51} & \ding{51}& \ding{51}& \ding{51}&\ding{51} &\ding{51} &\ding{55} &\ding{51} &\ding{51} &\ding{55} &\ding{51} & \ding{51}& \ding{55}& 0.0037& \textbf{0.0196} \\
    \hline
    \end{tabular}}
    \label{tab:case video}
\end{table*}

From Table \ref{tab:case video}, we observe that FRCOMA consistently achieves the lowest test MSE among all methods when training MSE is controlled, demonstrating superior generalization performance. 
As seen in simulations, FoF-SC performs the worst due to its oversimplified scalar coefficient assumption, which fails to capture intra-mode correlations and leads to poor training and test performance. FoF and MFPCA, limited by their linear modeling frameworks, also show lower predictive accuracy compared to FRCOMA. MToT and KMTPLS do not model the inherent data continuity, which reduces their effectiveness in this context. KMTPLS, benefiting from its nonlinear structure, performs relatively better among the baselines.
In terms of variable selection, FRCOMA identifies covariates with negligible predictive contribution while maintaining high prediction accuracy. Notably, variables such as color saturation, contrast, and edge concentration are consistently filtered out, indicating they do not meaningfully contribute to predicting trailer heat curves.

\subsection{Climate Dataset}
We further validate the effectiveness of FRCOMA using the ERA5 monthly averaged single-level dataset from the Copernicus Climate Data Store\footnote{\url{https://cds.climate.copernicus.eu/cdsapp\#!/dataset/reanalysis-era5-single-levels-monthly-means}}. As discussed in Section~\ref{sec:intro}, climate variables are naturally represented as continuous multi-way arrays, and nonlinear modeling is critical to capture the complex relationships inherent in atmospheric processes—making linear models insufficient. Variable selection is also necessary due to potential redundancy among covariates.

We select 10 continuous two-way array covariates ($p = 10, d_l = 2$ for all $l$), including water vapour (tcwv), UV radiation (uvb), cloud ice and liquid water (tciw, tclw), temperature measures (d2m, t2m), and integral measures (p53, p54, p79, p80). The response is also a continuous two-way array: total column supercooled liquid water (tclsw). Detailed variable definitions can be found on the ERA5 dataset page.

We consider three region-specific experiments using data from Pacific ($120^\circ W-150^\circ W, 0^\circ-30^\circ N$), Atlantic ($0^\circ - 30^\circ W, 30^\circ S - 60^\circ S$), and Indian Ocean ($60^\circ E-90^\circ E, 0^\circ-30^\circ S$). For each region, we use January and July (1998–2022) data for training ($n_{train} = 50$) and randomly select 20 samples from February and August for testing ($n_{test} = 20$). All functions are centralized and smoothed prior to modeling. FRCOMA uses a Gaussian kernel, $g_l(x_i^{(l)},x_j^{(l)}) = \exp(-\frac{||x_{i}^{(l)}-x_j^{(l)} ||^2_{\cX_l}}{\sigma_g^2}), l = 1,\dots,p$, and the projection operator $T$  maps onto a tensor product of trigonometric bases: $Tu(s,t) = \sum_{h = 0}^9\sum_{l=0}^9 \langle b_{hl},u\rangle_{\cY}b_{hl}(s,t)$. Here $\{b_{hl}(s,t)=\tilde{b}_{h}(s)\tilde{b}_{l}(t)|h = 0,\dots,9,l = 0,\dots,9\}$ where $\tilde{b}_l(t) = sin(2\pi k t)$ when $l = 2k-1$, $\tilde{b}_l(t) = cos(2\pi k t)$ when $l = 2k$. 
The hyperparameters $\lambda_1, \lambda_2, \sigma_g$ are selected via 5-fold cross-validation. 

For FoF, FoF-SC, and MFPCA, two-way arrays are vectorized into one-way functions as in Section~\ref{sec:simulation}. Their parameters are tuned to match FRCOMA’s training MSE for fair comparison. MToT and KMTPLS are applied after discretizing each mode into 100 grid points, with their decomposition ranks similarly adjusted. MToT-VC is excluded due to unknown covariate sparsity. The test MSEs and variable selection results are reported in Table~\ref{tab:case}.

\begin{table*}[hbt]
    \centering
    \caption{Variable selection and fitting results of Climate Dataset}
    \scalebox{0.7}{
    \begin{tabular}{llllllllllllll}
    \hline
    Region & Method & tcwv & uvb &tciw &tclw &d2m&t2m&p79&p54&p53&p80& $MSE_{train}$ & $MSE_{test}$\\
    \hline
    \multirow{6}{*}{Pacific}& MFPCA & / &/&/&/&/&/&/&/&/&/&0.3901 &0.9764 \\
    ~& FoF-SC &\ding{51} &\ding{55} &\ding{51} &\ding{51} &\ding{51} &\ding{51} &\ding{51} & \ding{51}& \ding{51}& \ding{51}& 0.8045& 1.9837\\
    ~& FoF &\ding{51} &\ding{51} &\ding{51} &\ding{51} &\ding{51} &\ding{51} &\ding{51} & \ding{51}& \ding{51}& \ding{51}& 0.3794& 0.9527\\
    ~ & MToT & / &/&/&/&/&/&/&/&/&/&0.3604 &0.7541 \\
    ~ & KMTPLS & / &/&/&/&/&/&/&/&/&/&0.3758 &0.7815 \\
    ~& FRCOMA &\ding{51} &\ding{55} &\ding{51} &\ding{51} &\ding{51} &\ding{51} &\ding{51} & \ding{51}& \ding{51}& \ding{55}& 0.3869& \textbf{0.5630}\\
    \hline
    \multirow{6}{*}{Atlantic}& MFPCA & / &/&/&/&/&/&/&/&/&/&0.3116 &0.9746 \\
    ~& FoF-SC &\ding{51} &\ding{51} &\ding{51} &\ding{51} &\ding{51} &\ding{51} &\ding{51} & \ding{51}& \ding{55}& \ding{51}& 0.7146& 1.7237\\
    ~& FoF &\ding{51} &\ding{55} &\ding{51} &\ding{51} &\ding{51} &\ding{51} &\ding{51} & \ding{51}& \ding{51}& \ding{51}& 0.3351& 1.0279\\
    ~ & MToT & / &/&/&/&/&/&/&/&/&/&0.3155 &0.7032 \\
    ~ & KMTPLS & / &/&/&/&/&/&/&/&/&/&0.3089 &0.6791 \\
    ~& FRCOMA &\ding{51} &\ding{55} &\ding{51} &\ding{55} &\ding{51} &\ding{51} &\ding{51} & \ding{51}& \ding{51}& \ding{51}& 0.3264& \textbf{0.6369}\\
    \hline
    \multirow{6}{*}{Indian Ocean}& MFPCA & / &/&/&/&/&/&/&/&/&/&0.3805 &1.1369 \\
    ~& FoF-SC &\ding{51} &\ding{55} &\ding{51} &\ding{51} &\ding{51} &\ding{51} &\ding{51} & \ding{51}& \ding{51}& \ding{51}& 0.7913& 1.6793\\
    ~& FoF &\ding{51} &\ding{55} &\ding{51} &\ding{51} &\ding{51} &\ding{51} &\ding{51} & \ding{51}& \ding{51}& \ding{51}& 0.3965& 1.1138\\
    ~ & MToT & / &/&/&/&/&/&/&/&/&/&0.4053 &0.8195 \\
    ~ & KMTPLS & / &/&/&/&/&/&/&/&/&/&0.3973 &0.8305 \\
    ~& FRCOMA &\ding{51} &\ding{55} &\ding{51} &\ding{51}&\ding{51} &\ding{51} &\ding{51} & \ding{51}& \ding{51}& \ding{55}& 0.4064& \textbf{0.7791}\\
    \hline
    \end{tabular}
    }
    \label{tab:case}
\end{table*}

From Table~\ref{tab:case}, we observe that FRCOMA achieves the lowest test MSE among all methods when training MSE is controlled. This superior performance stems from FRCOMA’s ability to preserve both the continuity and spatial structure of continuous multi-way climate data.
In contrast, function-on-function methods (FoF, FoF-SC, MFPCA) suffer from flattening two-way arrays into one-way curves, thereby losing essential spatial information. FoF-SC performs worst overall due to its use of scalar coefficients, which fail to capture interactions across spatial domains and cannot effectively reduce training error. Similarly, while MToT and KMTPLS support tensor inputs, they do not account for the continuity of the data, resulting in lower predictive accuracy than FRCOMA.

FRCOMA also provides meaningful variable selection. Notably, uvb (UV radiation) is discarded across all regions, consistent with its relative stability in equatorial zones and weak influence in higher latitudes. tclw (cloud liquid water) and p80 (ice-related flux divergence) exhibit different behaviors across the regions. In the Atlantic (mid-latitude) region, tclw is removed, likely due to the prevalence of ice clouds. Conversely, in the Pacific and Indian Ocean (low-latitude) regions, p80 is excluded, reflecting the rarity of ice clouds in warmer zones. These findings are consistent with known meteorological patterns \citep{lee2010characterization,hong2015characteristics}.

\section{CONCLUSION}
\label{sec:conslusion}

In this paper, we proposed FRCOMA, a scalable operator-valued kernel framework for nonlinear regression with heterogeneous continuous multi-way covariates and responses. By modeling the problem in a functional RKHS, FRCOMA avoids fixed-grid tensor representations and directly captures continuous-domain relationships. Separable operator-valued kernels yield a Kronecker-structured system for efficient spectral computation, while sparse kernel weights enable automatic covariate selection. We also developed a block coordinate descent algorithm and established guarantees for the representer form, support recovery, and convergence to stationary solutions.

Several directions remain for future research. First, extending FRCOMA beyond squared-error loss and Gaussian-type errors to quantile, heavy-tailed, or distributionally robust settings would improve robustness to noisy scientific data. Second, finer-grained selection at the mode or region level could provide greater interpretability than the current covariate-level mechanism. Finally, extending the framework to irregularly sampled, partially observed, or asynchronous multi-way functional data would broaden its applicability to real-world sensing and monitoring problems.

\begin{footnotesize}
    \bibliography{ref.bib}
\end{footnotesize}

\newpage

 \begin{appendix}
\renewcommand{\theequation}{A\arabic{equation}}
\renewcommand{\thealgorithm}{A\arabic{algorithm}}
\setcounter{section}{0}
\setcounter{subsection}{0}
\setcounter{equation}{0}
\setcounter{figure}{0}
\setcounter{algorithm}{0}
\renewcommand{\thetable}{A\arabic{table}}
\setcounter{table}{0} 

\section*{Appendix 1 - Proof of Theorem \ref{thm:representer}}
We use the notation $D_F$ for the Frechet derivative operator. Let $J(\Bf) = \sum_{j = 1}^{n}||y_j - \Bf(\bx_j)^{T} \symbfit{\alpha} ||^2_{\cY} + \lambda_1 \sum_{l = 1}^{p} \theta_l || f_{l}||_{\cF_l}^2$ be the function to be minimized. $\hat{\Bf} = (\hat{f_1},\dots,\hat{f_p})^T$ is the operator in $\cF = \cF_1 \times \cdots \times \cF_p$ such that $\hat{\Bf} = \arg\min_{\Bf \in \cF} J(\Bf) \Rightarrow D_FJ(\hat{\Bf}) = \mathbf{0} \Rightarrow D_FJ_l(\hat{f_l}) = 0, l = 1,\dots,p$.
 $J_l(f_l) = \sum_{j=1}^n G^l_j(f_l) + \lambda_1\theta_l L^l(f_l) + C$, where $C = \lambda_1 \sum_{h\neq l} \theta_h ||\hat{f_h} ||^2_{\cF_h}$ is independent of $f_l$, $L^l(f_l) = ||f_l||_{\cF_l}^2$ and $G_j^l(f_l) = ||y_j-\sum_{h\neq l}\alpha_h \hat{f_h}(x_j^{(h)}) - \alpha_l f_l(x_j^{(l)}) ||_\cY^2$.

To compute $D_FJ_l(\hat{f_l})$, we use Gateaus derivative $D_G$ of $J_l$ with respect to $f_l$ in the direction $H_l$, which is defined by,
\begin{equation*}
    D_GJ_l(f_l,H_l) = \lim_{\tau\rightarrow 0} \frac{J_l(f_l +\tau H_l) - J_l(f_l)}{\tau}.
\end{equation*}
Using the fact that $D_GJ_l(f_l,H_l) = \langle D_GJ(f_l),H_l \rangle_{\cF_l}$ we have,

\begin{equation*}
\begin{aligned}
 \text{(i)} \quad   \lim_{\tau\rightarrow 0}& \frac{||f_l + \tau H_l ||_{\cF_l}^2 - ||f_l||_{\cF_l}^2}{\tau} = 2\langle f_l,H_l\rangle_{\cF_l} \Rightarrow D_GL^l(f_l) = 2f_l,\\
 \text{(ii)} \quad \lim_{\tau\rightarrow 0}& \frac{||y_j-\sum_{h\neq l}\alpha_h \hat{f_h}(x_j^{(h)}) - \alpha_l f_l(x_j^{(l)}) -\tau\alpha_l H_l(x_j^{(l)})||_\cY^2-||y_j-\sum_{h\neq l}\alpha_h \hat{f_h}(x_j^{(h)}) - \alpha_l f_l(x_j^{(l)}) ||_\cY^2}{\tau} \\
 &= -2\alpha_l\langle y_j-\sum_{h\neq l}\alpha_h \hat{f_h}(x_j^{(h)}) - \alpha_l f_l(x_j^{(l)}), H_l(x_j^{(l)}) \rangle_\cY \\
 &=-2\alpha_l\langle K_l(x_j^{(l)},\cdot)(y_j-\sum_{h\neq l}\alpha_h \hat{f_h}(x_j^{(h)}) - \alpha_l f_l(x_j^{(l)})), H_l \rangle_{\cF_l} \\
 & = -2\alpha_l\langle K_l(x_j^{(l)},\cdot)\Tilde{u}_j^{(l)}, H_l \rangle_{\cF_l}\  \text{with}\quad \Tilde{u}_j^{(l)} = y_j-\sum_{h\neq l}\alpha_h \hat{f_h}(x_j^{(h)}) - \alpha_l f_l(x_j^{(l)})\\
 &\Rightarrow D_GG_j^l(f_l) = -2\alpha_l K_l(x_j^{(l)},\cdot)\Tilde{u}_j^{(l)}.
\end{aligned}
\end{equation*}
When the kernel $K_l$ is Mercer, Corollary 4.1.1 in \cite{kurdila2006convex} can be applied to show that $J_l$ is Frechet differentiable and that, $\forall f \in \cF_l, D_FJ_l(f) = D_GJ_l(f)$. Combine (i) and (ii) and $D_FJ_l(\hat{f_l}) = 0$. We have 
\begin{equation*}
    \hat{f_l} = \sum_{j=1}^n\frac{\alpha_l}{\theta_l \lambda_1}K_l(x_j^{(l)},\cdot)\Tilde{u}_j^{(l)}, l = 1,\dots,p.
\end{equation*}
Define $ u_j^{(l)} = \frac{\alpha_l}{\theta_l \lambda_1}\Tilde{u}_j^{(l)}$. We can get Theorem \ref{thm:representer}.

\section*{Appendix 2 - Proof of Proposition \ref{thm:only u}}

Take Gateaux derivative of target function in Eq. (\ref{eq:model 2}) with respect to $u_i^{(l)}, i = 1,\dots,n; l = 1,\dots,p$ and set it to 0. We have that $\hat{u}_i^{(l)}$ satisfies the following equation,
\begin{equation*}
    y_i = \sum_{l=1}^p\sum_{j=1}^n \theta_lK_l(x_j^{(l)},x_i^{(l)})\hat{u}_j^{(l)} + \lambda_1 \frac{\theta_l}{\alpha_l} \hat{u}_i^{(l)}.
\end{equation*}
When $\alpha_l = \theta_l,l = 1,\dots,p$ holds, we can get $\hat{u}_i^{(1)} = \cdots = \hat{u}_i^{(p)}$.

\section*{Appendix 3 - Proof of Theorem \ref{thm:u}}

According to Assumption \ref{ass:inverse operator}, we have that $T$ is finite rank operator. Without loss of generality, we assume operator $T$ has rank $\kappa$. Combining Algorithm \ref{alg:solve u}, we can write the explicit formulation of $\Tilde{\bu} = (\mathbf{K}_{\btheta^{*}} + \lambda_1\bI)^{-1}\mathbf{y}$ as 
\begin{equation*}
    \Tilde{\bu} = (\mathbf{K}_{\btheta^{*}} + \lambda_1\bI)^{-1}\mathbf{y}  = \sum_{b=1}^n\sum_{q=1}^\kappa[\frac{1}{\beta_b\delta_q + \lambda_1}\sum_{j=1}^n\langle v_{bj}w_q,y_j\rangle_{\cY}\mathbf{v}_b w_q].
\end{equation*}

As the operator $T$ is compact, it is not invertible. Yet we can define its pseudoinverse $T^{\dagger}$ which has been discussed in \cite{tsai2023latent}. Similarly, we can have $\mathbf{K}_{\btheta^{*}}^{\dagger} = G_{\btheta^{*}}^{\dagger}\otimes T^{\dagger}$, i.e.,
\begin{equation*}
    \mathbf{K}_{\btheta^{*}}^{\dagger}\mathbf{y} =   \sum_{b=1}^n\sum_{q=1}^\kappa[\frac{1}{\beta_b\delta_q}\sum_{j=1}^n\langle v_{rj}w_q,y_j\rangle_{\cY}\mathbf{v}_b w_q],
\end{equation*}
where $v_{bj}$ denotes the $j$-th element of the eigen vector $\mathbf{v}_b$.

We define $\mathring{\bu}$ as the solution to Eq. (\ref{eq:solve u}) when $\lambda_1 = 0$ and $\mathring{\bu} = \mathbf{K}^{\dagger}\mathbf{y}$. We can also find that the true coefficient function $\bu^{*}$ is the solution to Eq. (\ref{eq:solve u}) with $\lambda_1 = 0$ and without considering the noise function $\epsilon_i$, i.e., $\bu^{*} = \mathbf{K}_{\btheta^{*}}^{\dagger}(\mathbf{y} - \symbfit{\epsilon})$ with $\symbfit{\epsilon} = [\epsilon_1,\dots,\epsilon_n]^{T}$. Then we have,
\begin{equation}
\label{eq:thm2 1}
    \bbE\left[||\frac{1}{n}\sum_{i=1}^n \Tilde{u}_i - \frac{1}{n}\sum_{i=1}^n u_i^{*}||_{\cY}^2 \right] \leq 2\bbE\left[||\frac{1}{n}\sum_{i=1}^n \Tilde{u}_i - \frac{1}{n}\sum_{i=1}^n \mathring{u}_i||_{\cY}^2 \right] + 2\bbE\left[||\frac{1}{n}\sum_{i=1}^n \mathring{u}_i - \frac{1}{n}\sum_{i=1}^n u_i^{*}||_{\cY}^2 \right],
\end{equation}
where $\Tilde{u}_i,\mathring{u}_i$ and $u_i^{*}$ is the $i$-th element of $\Tilde{\bu}, \mathring{\bu}$ and $\bu^{*}$, respectively. Then we discuss the convergence of the two terms on the right-hand side separately. 

We talk about the first term $ \bbE\left[||\frac{1}{n}\sum_{i=1}^n \Tilde{u}_i -  \frac{1}{n}\sum_{i=1}^n \mathring{u}_i||_{\cY}^2 \right] $, 
\begin{equation*}
    \begin{aligned}
    \bbE&\left[||\frac{1}{n}\sum_{i=1}^n \Tilde{u}_i -  \frac{1}{n}\sum_{i=1}^n \mathring{u}_i||_{\cY}^2 \right] 
    = \frac{1}{n^2}\bbE\left[||\sum_{i=1}^n\sum_{b=1}^n\sum_{q=1}^\kappa[(\frac{1}{\beta_b\delta_q} - \frac{1}{\beta_b\delta_q+\lambda_1})\sum_{j=1}^n\langle v_{bj} w_q, y_j \rangle_{\cY} v_{bi} w_q]||_{\cY}^2 \right]\\
    &\leq \frac{1}{n^2}(\frac{\lambda_1}{b_{min}\delta_{min}(b_{min}\delta_{min}+\lambda_1)})^2\bbE\left[||\sum_{i=1}^n\sum_{b=1}^n\sum_{q=1}^\kappa[\sum_{j=1}^n\langle v_{bj} w_q, y_j \rangle_{\cY} v_{bi} w_q]||_{\cY}^2 \right]\\
    &\leq \frac{1}{n^2}(\frac{\lambda_1}{b_{min}\delta_{min}(b_{min}\delta_{min}+\lambda_1)})^2 n \sum_{i=1}^n \bbE\left[||\sum_{b=1}^n\sum_{q=1}^\kappa[\sum_{j=1}^n\langle v_{bj} w_q, y_j \rangle_{\cY} v_{bi} w_q]||_{\cY}^2 \right]\\
    &\leq \frac{1}{n}(\frac{\lambda_1}{b_{min}\delta_{min}(b_{min}\delta_{min}+\lambda_1)})^2  \sum_{i=1}^n \bbE\left[|| y_i||_{\cY}^2\right],
    \end{aligned}
\end{equation*}
where $b_{min}$ is the constant in Assumption \ref{ass:inverse operator} and $\delta_{min}$ is the minimum eigenvalue of $T$. Since we assume $\bbE[||y_i||_{\cY}^2] < \infty$ and $\lambda_1 = o(1)$, we have 
\begin{equation}
\label{eq:thm2 2}
    \bbE\left[||\frac{1}{n}\sum_{i=1}^n \Tilde{u}_i -  \frac{1}{n}\sum_{i=1}^n \mathring{u}_i||_{\cY}^2 \right] =  O(\lambda_1^2) \to 0, \ \text{as}\ n \to \infty.
\end{equation}

Before we discuss the second term, we give the following properties. 
By Assumption \ref{ass:bounded_y}, $\epsilon_1,\ldots,\epsilon_n$ are independent
mean-zero $\cY$-valued random elements with covariance operator
$\Gamma_\epsilon$. That is, for any $f,g\in \cY$,
\[
\bbE\{\langle \epsilon_i,f\rangle_\cY \langle \epsilon_i,g\rangle_\cY\}
=
\langle \Gamma_\epsilon f,g\rangle_\cY .
\]
In particular, for any $w_q\in \cY$,
\[
\bbE\{\langle \epsilon_i,w_q\rangle_\cY\}=0,
\qquad
\bbE\{\langle \epsilon_i,w_q\rangle_\cY^2\}
=
\langle \Gamma_\epsilon w_q,w_q\rangle_\cY
\le
\|\Gamma_\epsilon\|_{op}\,\|w_q\|_\cY^2 .
\]
Since $\{w_q\}_{q=1}^\kappa$ is orthonormal, $\|w_q\|_\cY=1$.

Let
\[
a_b:=\sum_{i=1}^n v_{bi},
\qquad
\xi_{bq}:=\sum_{j=1}^n v_{bj}\langle w_q,\epsilon_j\rangle_\cY .
\]
Then
\[
\frac{1}{n}\sum_{i=1}^n \mathring{u}_i-\frac{1}{n}\sum_{i=1}^n u_i^*
=
\frac{1}{n}\sum_{b=1}^n\sum_{q=1}^\kappa
\frac{a_b}{\beta_b\delta_q}\,\xi_{bq}\,w_q .
\]
Therefore,
\begin{equation*}
\begin{aligned}
\bbE\left\|
\frac{1}{n}\sum_{i=1}^n \mathring{u}_i-\frac{1}{n}\sum_{i=1}^n u_i^*
\right\|_\cY^2
&=
\frac{1}{n^2}
\bbE\left\|
\sum_{b=1}^n\sum_{q=1}^\kappa
\frac{a_b}{\beta_b\delta_q}\,\xi_{bq}\,w_q
\right\|_\cY^2 
&=
\frac{1}{n^2}
\sum_{q=1}^\kappa
\sum_{b=1}^n\sum_{b'=1}^n
\frac{a_b a_{b'}}{\beta_b\beta_{b'}\delta_q^2}
\bbE(\xi_{bq}\xi_{b'q}),
\end{aligned}
\end{equation*}
where we used the orthonormality of $\{w_q\}_{q=1}^\kappa$.

Next, by independence of $\epsilon_1,\ldots,\epsilon_n$,
\begin{equation*}
\begin{aligned}
\bbE(\xi_{bq}\xi_{b'q})
&=
\bbE\left[
\left(\sum_{j=1}^n v_{bj}\langle w_q,\epsilon_j\rangle_\cY\right)
\left(\sum_{j'=1}^n v_{b'j'}\langle w_q,\epsilon_{j'}\rangle_\cY\right)
\right] =
\sum_{j=1}^n\sum_{j'=1}^n
v_{bj}v_{b'j'}
\bbE\{\langle w_q,\epsilon_j\rangle_\cY
\langle w_q,\epsilon_{j'}\rangle_\cY\} \\
&=
\sum_{j=1}^n
v_{bj}v_{b'j}\,
\langle \Gamma_\epsilon w_q,w_q\rangle_\cY =
\delta_{bb'}\,
\langle \Gamma_\epsilon w_q,w_q\rangle_\cY .
\end{aligned}
\end{equation*}
Hence,
\begin{equation*}
\begin{aligned}
\bbE\left\|
\frac{1}{n}\sum_{i=1}^n \mathring{u}_i-\frac{1}{n}\sum_{i=1}^n u_i^*
\right\|_\cY^2
&=
\frac{1}{n^2}
\sum_{q=1}^\kappa
\sum_{b=1}^n
\frac{a_b^2}{\beta_b^2\delta_q^2}
\langle \Gamma_\epsilon w_q,w_q\rangle_\cY 
&\le
\frac{1}{n^2}\frac{1}{b_{\min}^2\delta_{\min}^2}
\left(\sum_{b=1}^n a_b^2\right)
\left(\sum_{q=1}^\kappa
\langle \Gamma_\epsilon w_q,w_q\rangle_\cY\right).
\end{aligned}
\end{equation*}
Since $a_b=\langle v_b,\mathbf 1_n\rangle$ and $\{v_b\}_{b=1}^n$ is an orthonormal basis of $\mathbb R^n$,
\[
\sum_{b=1}^n a_b^2
=
\sum_{b=1}^n \langle v_b,\mathbf 1_n\rangle^2
=
\|\mathbf 1_n\|_2^2
=
n.
\]
Therefore,
\[
\bbE\left\|
\frac{1}{n}\sum_{i=1}^n \mathring{u}_i-\frac{1}{n}\sum_{i=1}^n u_i^*
\right\|_\cY^2
\le
\frac{1}{n}\frac{1}{b_{\min}^2\delta_{\min}^2}
\sum_{q=1}^\kappa
\langle \Gamma_\epsilon w_q,w_q\rangle_\cY .
\]
Finally, since
\[
\langle \Gamma_\epsilon w_q,w_q\rangle_\cY
\le
\|\Gamma_\epsilon\|_{op}\|w_q\|_\cY^2
=
\|\Gamma_\epsilon\|_{op},
\]
we obtain
\[
\bbE\left\|
\frac{1}{n}\sum_{i=1}^n \mathring{u}_i-\frac{1}{n}\sum_{i=1}^n u_i^*
\right\|_\cY^2
\le
\frac{\kappa\|\Gamma_\epsilon\|_{op}}
{n\,b_{\min}^2\delta_{\min}^2}.
\]

Then we can get,
\begin{equation}
\label{eq:thm2 3}
    \bbE\left[||\frac{1}{n}\sum_{i=1}^n \mathring{u}_i - \frac{1}{n}\sum_{i=1}^n u_i^{*}||_{\cY}^2 \right] = O(n^{-1}) \to  0, \ \text{as}\ n \to \infty.
\end{equation}

Combine Eq. (\ref{eq:thm2 1}), Eq. (\ref{eq:thm2 2}), Eq. (\ref{eq:thm2 3}) and $\lambda_1 = \omega(n^{\frac{c_1-1}{2}})$. We complete the proof of Theorem \ref{thm:u}. 

\section*{Appendix 4 - Proof of Theorem \ref{thm:theta}}

Without considering the nonnegative constraint of $\btheta$, the target function in Eq. (\ref{eq:solve theta}) can be rewritten as follows,
\begin{equation*}
\begin{aligned}
    h_{\bu^{*}}(\btheta) &= s_{\bu^{*}}(\btheta) + \lambda_2\sum_{l=1}^p\theta_l 
    =\sum_{i=1}^n||y_i - \mathbf{p}_i^T\btheta ||_{\cY}^2 + \sum_{i=1}^n \mathbf{q}_i^T\btheta + \lambda_2\sum_{l=1}^p|\theta_l|,
\end{aligned} 
\end{equation*}
where $\mathbf{p}_i \in \bbR^p$ and the $l$-th element of $\mathbf{p}_i$ is $\sum_{j=1}^nK_l(x_j^{(l)},x_i^{(l)})u_j^{*}$,  $\mathbf{q}_i \in \bbR^p$ and the $l$-th element of $\mathbf{q}_i$ is $\lambda_1\sum_{j=1}^n\langle K_l(x_j^{(l)},x_i^{(l)})u_j^{*},u_i^{*} \rangle_{\cY}$. We also have $\Tilde{\btheta} = \arg\min_{\btheta \in \bbR^p} h_{\bu^{*}}(\btheta)$ and $y_i = \bp_i^T\btheta^{*} + \epsilon_i$ where $\epsilon_i$ is the white noise function. 

Define $\hat{\bgamma} = \Tilde{\btheta} - \btheta^{*}$ and 
\begin{equation*}
    v(\bgamma) = \sum_{i=1}^n||\epsilon_i - \bp_i^T\bgamma ||_{\cY}^2 + \sum_{i=1}^n\bq_i^T(\bgamma+\btheta^{*}) + \lambda_2\sum_{l=1}^p|\gamma_l+\theta_l^{*}|,
\end{equation*}
where $\gamma_l$ is the $l$-th element of $\bgamma$ and we have $\hat{\bgamma} = \arg\min_{\bgamma\in \bbR^{p}}v(\bgamma)$. Note that we can take derivatives of the first two terms to obtain,
\begin{equation}
\label{eq:thm3 1}
\frac{d[\sum_{i=1}^n||\epsilon_i - \bp_i^T\bgamma ||_{\cY}^2 + \sum_{i=1}^n\bq_i^T(\bgamma+\btheta^{*})]}{d\bgamma} = 2n\bC_{\bu^{*}}\bgamma - 2\sqrt{n}\bW_{\bu^{*}} + n\lambda_1\bR_{\bu^{*}},
\end{equation}
where $\bC^n_{\bu^{*}} \in \bbR^{p\times p}$ with the $lh$-th element as $\frac{1}{n}\sum_{i=1}^n\langle \sum_{j=1}^n K_l(x_j^{(l)},x_i^{(l)})u_j^{*}, \sum_{j=1}^n K_h(x_j^{(h)},x_i^{(h)})u_j^{*}\rangle_\cY$, $\bW^n_{\bu^{*}} \in \bbR^{p}$ with the $l$-th element as $\frac{1}{\sqrt{n}}\sum_{i=1}^n\langle \sum_{j=1}^n K_l(x_j^{(l)},x_i^{(l)})u_j^{*}, \epsilon_i \rangle_{\cY}$, $\bR^n_{\bu^{*}} \in \bbR^{p}$ with the $l$-th element as $\frac{1}{n}\sum_{i,j=1}^n\langle K_l(x_j^{(l)},x_i^{(l)})u_j^{*}, u_i^{*} \rangle_{\cY}$.

According to the property of the Lasso estimator, we have 
\begin{equation}
\label{eq:thm3 2}
    \begin{aligned}
        \frac{d[m(\btheta)]}{d\theta_j}|_{\theta_j = \Tilde{\theta}_j} = \lambda_2 sign(\Tilde{\theta}_j) \ \text{for} \ j \ \text{s.t.} \ \Tilde{\theta}_j \neq 0, \\
        |\frac{d[m(\btheta)]}{d\theta_j}|_{\theta_j = \Tilde{\theta}_j}| \leq \lambda_2  \ \text{for} \ j \ \text{s.t.} \ \Tilde{\theta}_j = 0.
    \end{aligned}
\end{equation}

Without loss of generality, we assume that the first
$r$ elements of $\btheta^{*}$, i.e., $\btheta^{*}(1)$, are greater than 0, and the remaining $p-r$ elements, i.e., $\btheta^{*}(2)$, are equal to 0. Let $\bW^n_{\bu^{*}}(1)$,$\bR^n_{\bu^{*}}(1)$,$\hat{\bgamma}(1)$ and $\bW^n_{\bu^{*}}(2)$, $\bR^n_{\bu^{*}}(2)$,$\hat{\bgamma}(2)$ denote the first $r$ and the last $p-r$ entries of $\bW^n_{\bu^{*}}$, $\bR^n_{\bu^{*}}$ and $\hat{\bgamma}$ respectively.
Combine Eq. (\ref{eq:thm3 1}), Eq. (\ref{eq:thm3 2}) and the uniqueness of the Lasso solution. If there exists $\hat{\bgamma}$ that satisfies the following, 
\begin{equation*}
    \begin{aligned}
    \bC^n_{\bu^{*},11}(\sqrt{n}\hat{\bgamma}(1)) - \bW^n_{\bu^{*}}(1) + \frac{\lambda_1\sqrt{n}}{2}\bR^n_{\bu^{*}}(1) &= -\frac{\lambda_2}{2\sqrt{n}} sign(\btheta^{*}(1)),\\
    -\frac{\lambda_2}{2\sqrt{n}}\mathbf{1} \leq \bC^n_{\bu^{*},21}(\sqrt{n}\hat{\bgamma}(1)) - \bW^n_{\bu^{*}}(2) + \frac{\lambda_1\sqrt{n}}{2}\bR^n_{\bu^{*}}(2) &\leq \frac{\lambda_2}{2\sqrt{n}}\mathbf{1},\\
    |\hat{\bgamma}(1) | &\leq |\btheta^{*}(1)|,
\end{aligned}
\end{equation*}
then $sign(\hat{\btheta}(1)) = sign(\btheta^{*}(1))$ and $sign(\hat{\btheta}(2)) = \bgamma(2) = 0$.

Then the above equations can be implied by
\begin{equation}
    \label{eq:thm3 3}
    |(\bC^n_{\bu^{*},11})^{-1}\bW^n_{\bu^{*}}(1)| \leq \sqrt{n}\left(|\btheta^{*}(1)| - \frac{\lambda_2}{2n}|(\bC^n_{\bu^{*},11})^{-1}sign(\btheta^{*}(1))| - \lambda_1|(\bC^n_{\bu^{*},11})^{-1}\frac{\bR_{\bu^{*}}^n(1)}{2}|\right),
\end{equation}
and
\begin{equation}
    \begin{aligned}
    \label{eq:thm3 4}
    &|\bC^n_{\bu^{*}21}(\bC^n_{\bu^{*},11})^{-1}\bW^n_{\bu^{*}}(1)-\bW^n_{\bu^{*}}(2)|\\ 
    &\leq\left|\frac{\lambda_2}{2\sqrt{n}}(1-|\bC^n_{\bu^{*}21}(\bC^n_{\bu^{*},11})^{-1}sign(\btheta^{*}(1))) - \lambda_1\sqrt{n}|\bC^n_{\bu^{*},21}(\bC^n_{\bu^{*},11})^{-1}\frac{\bR^n_{\bu^{*}}(1)}{2} - \frac{\bR^n_{\bu^{*}}(2)}{2}| \right|.\\
\end{aligned}
\end{equation}

We define event $A_n$ as that Eq. (\ref{eq:thm3 3}) holds and event $B_n$ as that Eq. (\ref{eq:thm3 4}) holds. Then we can get, 
\begin{equation}
\label{eq:thm3 5}
    P(sign(\Tilde{\btheta}) = sign(\btheta^{*})) \geq P(A_n \cap B_n).
\end{equation}
Furthermore, we have
\begin{equation}
\begin{aligned}
\label{eq:thm3 6}
    1-P(A_n\cap B_n)&\leq P(A_n^c) + P(B_n^c) \\
    &\leq \sum_{i=1}^r P\left(|\eta_i| \geq \sqrt{n}\left( |\btheta^{*}(1)| - \frac{\lambda_2}{2n}|a_i| - \lambda_1|b_i| \right)\right) + 
    \sum_{i=1}^{p-r}P\left( |\zeta_i| \geq\left| \frac{\lambda_2}{2\sqrt{n}}\xi_{\bu^{*}} - \lambda_1\sqrt{n}|d_i|\right|\right),
\end{aligned}
\end{equation}
where $\symbfit{\eta} = (\eta_1,\dots, \eta_r)^T = (\bC^n_{\bu^{*},11})^{-1}\bW_{\bu^{*}}(1)$, $\symbfit{\zeta} = (\zeta_1,\dots,\zeta_{p-r})^T = \bC^n_{\bu^{*},21}(\bC^n_{\bu^{*},11})^{-1}\bW^n_{\bu^{*}}(1)-\bW^n_{\bu^{*}}(2)$, $(a_1,\dots,a_r)^T = (\bC^n_{\bu^{*},11})^{-1}sign(\btheta^{*}(1))$, $(b_1,\dots,b_r)^T = (\bC^n_{\bu^{*},11})^{-1}\frac{\bR^n_{\bu^{*}}(1)}{2}$ , $(d_1,\dots,d_{p-r})^T =\bC^n_{\bu^{*},21}(\bC^n_{\bu^{*},11})^{-1}\frac{\bR^n_{\bu^{*}}(1)}{2} - \frac{\bR^n_{\bu^{*}}(2)}{2} $, and $\xi_{\bu^{*}}$ has the same definition in Assumption \ref{ass:irrepresentable}.

Under Assumption \ref{ass:irrepresentable}, we have, 
\begin{equation*}
    \begin{aligned}
        (\bC^n_{\bu^{*},11})^{-1}\bW_{\bu^{*}}^n(1) &\overset{d}{\rightarrow} N(0,\bC_{\bu^{*},11}^{-1} ), \\
        \bC^n_{\bu^{*},21}(\bC^n_{\bu^{*},11})^{-1}\bW^n_{\bu^{*}}(1)-\bW^n_{\bu^{*}}(2) &\overset{d}{\rightarrow} N(0,\bC_{\bu^{*},22}-\bC_{\bu^{*},21}\bC_{\bu^{*},11}^{-1}\bC_{\bu^{*},12}).
    \end{aligned}
\end{equation*}

Therefore, all $\eta_i$ and $\zeta_i$ converge in distribution to Gaussian random variables with mean 0 and
finite variance $E(\eta_i)^2\leq s^2$ and $E(\zeta_i)^2 \leq s^2$ respectively for some constant $s \geq 0$.

Consider that Gaussian distribution has its tail probability bounded by 
\begin{equation*}
    1-\Phi(t) < t^{-1}e^{-\frac{1}{2}t^2}.
\end{equation*}
Under Assumption \ref{ass:lambda rate}, we have $\lambda_1 = o(1),\lambda_2 = o(n), \frac{\lambda_1}{\lambda_2} = o(\frac{1}{n})$ and $\lambda_2 = \Omega(n^{\frac{1+c_2}{2}})$ with $0\leq c_2 <1$, then we have 
\begin{equation}
\label{eq:thm3 7}
    \begin{aligned}
        \sum_{i=1}^r &P\left(|\eta_i| \geq \sqrt{n}\left( |\btheta(1)| - \frac{\lambda_2}{2n}|a_i| - \lambda_1|b_i| \right)\right)\\ &\leq (1+o(1))\sum_{i=1}^r\left(1 - \Phi\left((1+o(1))\frac{1}{s}\sqrt{n}|\btheta(1)| \right)\right) \asymp \frac{e^{-\frac{n}{2}}}{\sqrt{n}} = o(e^{-n^{c_2}}),
    \end{aligned}
\end{equation}
and 
\begin{equation}
\label{eq:thm3 8}
    \begin{aligned}
        \sum_{i=1}^{p-r}P\left( |\zeta_i| \geq\left| \frac{\lambda_2}{2\sqrt{n}}\xi_i - \lambda_1\sqrt{n}|d_i|\right|\right) 
        \leq \sum_{i=1}^{p-r} \left(1 - \Phi\left((1+o(1))\frac{1}{s}\xi_i\frac{\lambda_2}{2\sqrt{n}} \right) \right) \asymp \frac{\sqrt{n}}{\lambda_2}e^{-\frac{\lambda_2^2}{2n}} = o(e^{-n^{c_2}}),
    \end{aligned}
\end{equation}
where we say $a_n \asymp b_n$ if $\lim_{n\to\infty}a_n/b_n = 1$.

Combining Eq. (\ref{eq:thm3 5}), Eq. (\ref{eq:thm3 6}), Eq. (\ref{eq:thm3 7}) and Eq. (\ref{eq:thm3 8}), we can get Theorem \ref{thm:theta} without considering the nonnegative constraints. As the true parameter $\btheta^{*}$ is also positive, the nonnegative constraints will not influence the sign estimation of $\btheta^{*}$. Adding the nonnegative constraints, Theorem \ref{thm:theta} still holds. 

\section*{Appendix 5 - Proof of Theorem \ref{thm:theta global}}

For each $l=1,\dots,p$, define the linear operator $B_l:\cY^n\to \cY^n$ by
\[
[B_l(\bu)]_i:=\sum_{j=1}^n K_l(x_j^{(l)},x_i^{(l)})u_j,
\qquad i=1,\dots,n.
\]
Then the quantities used in the proof of Theorem~\ref{thm:theta} can be written as
\[
\bC_{\bu,lh}^n=\frac1n\langle B_l(\bu),B_h(\bu)\rangle_{\cY^n},
\qquad
\bW_{\bu,l}^n=\frac1{\sqrt n}\langle B_l(\bu),\epsilon\rangle_{\cY^n},
\qquad
\bR_{\bu,l}^n=\frac1n\langle B_l(\bu),\bu\rangle_{\cY^n},
\]
where $\epsilon=(\epsilon_1,\dots,\epsilon_n)^T$.

Since the sample points $\{x_i^{(l)}\}_{i=1}^n$ are fixed and each kernel $K_l$ is bounded on these points, $B_l$ is a bounded linear operator from $\cY^n$ to $\cY^n$. Hence,
\[
\|B_l(\hat \bu)-B_l(\bu^\ast)\|_n=o_p(1), \qquad l=1,\dots,p.
\]
It follows that
\[
\|\bC_{\hat \bu}^n-\bC_{\bu^\ast}^n\|_\infty=o_p(1),
\qquad
\|\bR_{\hat \bu}^n-\bR_{\bu^\ast}^n\|_\infty=o_p(1).
\]
Moreover, since
\[
|\bW_{\hat \bu,l}^n-\bW_{\bu^\ast,l}^n|
=
\frac1{\sqrt n}\big|\langle B_l(\hat \bu)-B_l(\bu^\ast),\epsilon\rangle_{\cY^n}\big|
\le
\|B_l(\hat \bu)-B_l(\bu^\ast)\|_n\,\|\epsilon\|_n,
\]
and $\|\epsilon\|_n=O_p(1)$, we also have
\[
\|\bW_{\hat \bu}^n-\bW_{\bu^\ast}^n\|_\infty=o_p(1).
\]
Therefore,
\begin{equation}
\label{eq:proof new 1}
\bC_{\hat \bu}^n=\bC_{\bu^\ast}^n+o_p(1),\qquad
\bW_{\hat \bu}^n=\bW_{\bu^\ast}^n+o_p(1),\qquad
\bR_{\hat \bu}^n=\bR_{\bu^\ast}^n+o_p(1).
\end{equation}

By Assumption \ref{ass:irrepresentable} at $\bu^\ast$, $\bC_{\bu^\ast,11}$ is positive definite. Hence, by standard matrix perturbation arguments,
\begin{equation}
    \label{eq:proof new 2}
(\bC_{\hat \bu,11}^n)^{-1}=(\bC_{\bu^\ast,11})^{-1}+o_p(1).
\end{equation}

Similarly, by Assumption 5 at $\bu^\ast$ and Eq.~(\ref{eq:proof new 1}),
$\|\bC_{\hat \bu,21}^n(\bC_{\hat \bu,11}^n)^{-1}sign(\btheta^\ast_{(1)})\|_\infty
\le 1-\frac{\xi_{\bu^\ast}}{2}$
with probability tending to one.

Now, since $(\hat{\bu},\hat{\btheta})$ is a global minimizer of Eq.~(11), $\hat{\btheta}$ minimizes $h_{\hat \bu}(\btheta)$ over $\mathbb R_+^p$. Therefore, the KKT conditions for $\hat{\btheta}$ take the form
\begin{equation}
\label{eq:proof new 3}
2n\,\bC_{\hat \bu}^n(\hat{\btheta}-\btheta^\ast)-2\sqrt n\,\bW_{\hat \bu}^n+n\lambda_1 \bR_{\hat \bu}^n+\lambda_2 z=0,
\end{equation}
where $z_{(1)}=sign(\hat{\btheta}_{(1)}),\qquad \|z_{(2)}\|_\infty\le 1$.

We first consider the active block. Partitioning Eq.~(\ref{eq:proof new 3}) according to $(1)$ and $(2)$ and solving the $(1)$-block gives
\begin{equation}
\label{eq:proof new 4}
\hat{\btheta}_{(1)}-\btheta^\ast_{(1)}
=
\frac1{\sqrt n}(\bC_{\hat \bu,11}^n)^{-1}\bW_{\hat \bu,(1)}
-\frac{\lambda_1}{2}(\bC_{\hat \bu,11}^n)^{-1}\bR_{\hat \bu,(1)}
-\frac{\lambda_2}{2n}(\bC_{\hat \bu,11}^n)^{-1}sign(\hat{\btheta}_{(1)}).
\end{equation}
By Eq.~(\ref{eq:proof new 1}), Eq.~(\ref{eq:proof new 2}), Assumption 6 at $\bu^\ast$, and Assumption 3, we have $\lambda_1=o(1),\qquad \lambda_2/n\to 0$ and
\[
\|(\bC_{\hat \bu,11}^n)^{-1}\bW_{\hat \bu,(1)}\|_\infty=O_p(1),
\qquad
\|(\bC_{\hat \bu,11}^n)^{-1}\bR_{\hat \bu,(1)}\|_\infty=O_p(1).
\]

Hence Eq.~(\ref{eq:proof new 4}) implies
\[
\|\hat{\btheta}_{(1)}-\btheta^\ast_{(1)}\|_\infty=o_p(1).
\]
Suppose $S$ is the set of active dimensions. Since $min_{l\in S} \theta_{l}>0$, we conclude that
\begin{equation}
\label{eq:proof new 5}
P\bigl(\hat\theta_l>0,\ \forall l\in S\bigr)\to 1.
\end{equation}

Next we consider the inactive block. Substituting Eq.~(\ref{eq:proof new 4}) into the $(2)$-block of Eq.~(\ref{eq:proof new 3}), we obtain
\begin{equation}
\label{eq:proof new 6}
    \begin{aligned}
        z_{(2)} =
\bC_{\hat \bu,21}^n(\bC_{\hat \bu,11}^n)^{-1}sign(\hat{\btheta}_{(1)})
+\frac{2\sqrt n}{\lambda_2}
\Bigl[
\bW_{\hat\bu,(2)}-\bC_{\hat \bu,21}^n(\bC_{\hat \bu,11}^n)^{-1}\bW_{\hat \bu,(1)}
\Bigr] \\
+\frac{n\lambda_1}{\lambda_2}
\Bigl[
\bC_{\hat \bu,21}^n(\bC_{\hat \bu,11}^n)^{-1}\bR_{\hat \bu,(1)}-\bR_{\hat \bu,(2)}
\Bigr].
    \end{aligned}
\end{equation}

By Eq.~(\ref{eq:proof new 5}), $sign(\hat{\btheta}_{(1)})=sign(\btheta^\ast_{(1)})$ with probability tending to one, and thus the first term in Eq.~(\ref{eq:proof new 6}) is bounded by $1-\xi_{\bu^\ast}/2$ with probability tending to one.

For the second term in Eq.~(\ref{eq:proof new 6}), Eq.~(\ref{eq:proof new 1}) implies that it equals
\[
\frac{2\sqrt n}{\lambda_2}
\Bigl[
\bW_{\bu^\ast,(2)}-\bC_{\bu^\ast,21}(\bC_{\bu^\ast,11})^{-1}\bW_{\bu^\ast,(1)}
\Bigr]
+o_p\!\left(\frac{\sqrt n}{\lambda_2}\right).
\]
Under Assumption ,\ref{ass:irrepresentable} the bracketed term is $O_p(1)$, while Assumption 3 gives $\lambda_2=\Omega\!\left(n^{\frac{1+c_2}{2}}\right)$,
hence $\frac{\sqrt n}{\lambda_2}=o(1)$. 
Therefore, the second term is $o_p(1)$.

For the third term in Eq.~(\ref{eq:proof new 6}), Eq.~(\ref{eq:proof new 1}) gives
\[
\frac{n\lambda_1}{\lambda_2}
\Bigl[
\bC_{\hat \bu,21}^n(\bC_{\hat \bu,11}^n)^{-1}\bR_{\hat \bu,(1)}-\bR_{\hat \bu,(2)}
\Bigr]
=
\frac{n\lambda_1}{\lambda_2}
\Bigl[
\bC_{\bu^\ast,21}(\bC_{\bu^\ast,11})^{-1}\bR_{\bu^\ast,(1)}-\bR_{\bu^\ast,(2)}
\Bigr]
+o_p\!\left(\frac{n\lambda_1}{\lambda_2}\right).
\]
By Assumption 3,
$\frac{\lambda_1}{\lambda_2}=o\!\left(\frac1n\right)$, 
and hence $\frac{n\lambda_1}{\lambda_2}=o(1)$. 
So the third term is also $o_p(1)$.

Combining the above bounds, Eq.~(\ref{eq:proof new 6}) yields
\[
\|z_{(2)}\|_\infty
\le
1-\frac{\xi_{\bu^\ast}}{2}+o_p(1)
<1
\]
with probability tending to one. By the KKT condition for the nonnegative Lasso-type problem, this implies 
$\hat{\btheta}_{(2)}=0$ 
with probability tending to one. Together with Eq.~(\ref{eq:proof new 5}), we conclude that
\[
P\bigl(sign(\hat{\btheta})=sign(\btheta^\ast)\bigr)\to 1.
\]
This completes the proof.

\section*{Appendix 6 - Proof of Proposition 2}
We show the convexity of the target functions in Eq.~(\ref{eq:solve u}) and Eq.~(\ref{eq:solve theta}) by verifying that their respective Hessian matrices are positive definite.

For Eq.~(\ref{eq:solve u}), the Hessian matrix is given by:
$\mathbf{K}_{\btheta}(\mathbf{K}_{\btheta}+ \lambda_1\bI)$ with $\mathbf{K}_{\btheta} = [K(\bx_i,\bx_j)]_{i,j = 1}^n$. According to Definition~\ref{def:op kernel}, $\mathbf{K}_{\btheta}$ is a positive definite operator-valued kernel matrix, and since $\lambda_1 >0$, the entire expression is semipositive definite, establishing convexity with respect to $\bu$. 

For Eq.~(\ref{eq:solve theta}), the Hessian matrix is $\Tilde{\mathbf{K}}_{\bu}$, whose $lh$-th element is $\sum_{i = 1}^n \langle\sum_{j = 1}^n K_l(x_j^{(l)},x_i^{(l)})u_j, \sum_{j = 1}^n K_h(x_j^{(h)},x_i^{(h)})u_j \rangle_{\cY}$. 

For any $\mathbf{b} = (b_1,\dots,b_p)^{T}\in \bbR^{p}$, we have 
\begin{equation*}
    \begin{aligned}
    \mathbf{b}^T \Tilde{\mathbf{K}}_{\bu}\mathbf{b} &= \sum_{h=1}^p\sum_{l=1}^p\sum_{i=1}^n b_h b_l\langle\sum_{j = 1}^n K_l(x_j^{(l)},x_i^{(l)})u_j, \sum_{j = 1}^n K_h(x_j^{(h)},x_i^{(h)})u_j \rangle_{\cY}\\& = 
    \sum_{i=1}^n\langle\sum_{j = 1}^n \sum_{l=1}^nb_lK_l(x_j^{(l)},x_i^{(l)})u_j, \sum_{j = 1}^n\sum_{h=1}^nb_h K_h(x_j^{(h)},x_i^{(h)})u_j \rangle_{\cY} \\
    &=\sum_{i=1}^n ||\sum_{j = 1}^n \sum_{l=1}^pb_lK_l(x_j^{(l)},x_i^{(l)})u_j ||_{\cY}^2 \geq 0.
\end{aligned}
\end{equation*} Then we have $\Tilde{\mathbf{K}}_{\bu}$ is semipositive definite, confirming that the objective function in Eq.~(\ref{eq:solve theta}) is convex with respect to $\btheta$. 

\section*{Appendix 7 - Proof of Theorem \ref{thm:BCD}}
We define an extended objective function $Q(\bu,\btheta) = q(\bu,\btheta) + r(\btheta)$ where $r(\btheta) = \begin{cases}0 , \btheta \geq 0 \\ \infty, o.w.\end{cases}$. This reformulation allows us to recast the original constrained optimization problem over $\cY^n \times \bbR^p_{+}$ as an unconstrained optimization of $Q(\bu,\btheta)$ over $\cY^n \times \bbR^p$.

From Proposition~\ref{prop:convex} and Assumption~\ref{ass:strongly convex}, we know that $q(\bu,\btheta)$ is block multi-convex and strongly convex with respect to each block variable. Furthermore, $Q(\bu,\btheta) \geq 0 >-\infty$ ensures that the objective is bounded from below. 

Therefore, by invoking Theorem 2.3 in \cite{xu2013block}, we conclude that the BCD algorithm described in Algorithm~\ref{alg:final BCD} converges to a stationary point of $Q(\bu,\btheta)$, and equivalently, to a stationary point of the original problem in Eq.~(\ref{eq:model 3}).

\section*{Appendix 8 - Computational Time Comparison}

To further evaluate the computational efficiency of FRCOMA, we compare its running time with all baseline methods under the simulation settings in Section~\ref{sec:simulation}. Specifically, we consider the G--M1 setting for the three simulation scenarios in the main paper, namely the one-way, multi-way-small, and multi-way-large scenarios. Here, G denotes the Gaussian kernel setting and M1 denotes the first sparsity setting used in the simulation study. For each scenario and each method, the experiment is independently repeated 10 times, and we report the average wall-clock running time in seconds. All methods are implemented and evaluated under the same computational environment.

\begin{table*}[htbp]
    \centering
    \caption{Average running time under the G--M1 simulation setting (seconds)}
    \label{tab:runtime-gm1}
    \begin{threeparttable}
    \begin{tabular}{lrrrrrrr}
        \toprule
        Scenario & MFPCA & FoF-SC & FoF & MToT & MToT-VC & KMTPLS & FRCOMA \\
        \midrule
        One-way          & 1.21 & 2.54  & 2.87  & 3.97  & 3.97  & 5.63  & 7.85  \\
        Multi-way-small  & 1.35 & 6.89  & 7.35  & 5.94  & 5.94  & 31.88 & 42.46 \\
        Multi-way-large  & 1.58 & 26.72 & 28.91 & 22.67 & 22.67 & 64.89 & 89.38 \\
        \bottomrule
    \end{tabular}
    \end{threeparttable}
\end{table*}

Table~\ref{tab:runtime-gm1} summarizes the average running time of all methods under the G--M1 setting. MFPCA is the fastest method across the three scenarios. This is mainly because its core procedure relies on local smoothing, eigendecomposition, and closed-form score estimation, rather than a nested iterative optimization scheme. As a result, its running time remains relatively small and stable from the one-way case to the multi-way cases.

The comparison among FoF-SC, FoF, MToT, and MToT-VC shows a clear dependence on the data structure. FoF-SC and FoF are designed for continuous one-way data, and therefore run faster than MToT and MToT-VC in the one-way scenario. However, for multi-way data, they require reshaping the observations into one-way representations, which substantially increases the computational burden. In contrast, MToT and MToT-VC exploit low-rank tensor decomposition to reduce the effective dimensionality, making them more efficient in the multi-way scenarios.

KMTPLS and FRCOMA require more computation because both methods use kernel-based modeling to capture nonlinear relationships. FRCOMA is slower than KMTPLS since it further solves the proposed operator-valued kernel learning problem through an outer BCD algorithm and an inner update of the kernel-weight vector $\btheta$. Nevertheless, the running time of FRCOMA remains moderate in all scenarios, and the additional cost is acceptable given its improved fitting accuracy and variable selection performance.

The empirical pattern is consistent with the complexity discussion in Remark~\ref{rem:computation}. FRCOMA involves scalar-kernel precomputation, Kronecker-structured spectral updates, and a second-level optimization for $\btheta$. Therefore, it is naturally more expensive than MFPCA and low-rank tensor methods, but this cost reflects the additional flexibility needed for nonlinear continuous multi-way regression.

\section*{Appendix 9 - Variable Explanation of Movie Trailer Dataset}

In this subsection, we describe the variables in Section \ref{sec:video}. We divided the 19 covariates into three categories, which are static visual features, dynamic visual features and sound features.
 
For static visual features, \textit{color} ($x^{(10)}$) measures how colorful each video frame is, following \cite{hasler2003measuring}. 
\textit{bright} ($x^{(15)}$) is the average of the RGB (Red, Green, Blue) channels to resemble human brightness perception.
\textit{contr} ($x^{(16)}$) is the 0.05 quantile range of the
brightness values in each video frame.
\textit{sat} ($x^{(13)}$) refers to the degree of purity or
intensity of colors, measured by the saturation values in the HSV (Hue, Saturation, and Value)
color space of an image. 
\textit{quality} ($x^{(11)}$) is proposed by \cite{ke2006design}, where high-quality images have regions
with sharp edges as well as regions without large differences in intensity. 
Extract sharp edges from the video frames and then calculate the \textit{edgeDes} ($x^{(17)}$), measured as the area occupied by edges, and the \textit{edgeDis} ($x^{(18)}$), measured as the
mean of the Euclidean distances among the edge points.  \textit{box} ($x^{(19)}$) is the minimum
size of a bounding box containing 95$\%$ of the edges, to measure how concentrated the
information-rich areas are in the images. 
\textit{colVar} ($x^{(12)}$) is the number of bins with high frequencies after sorting the hue
values in HSV space into discrete bins \citep{ke2006design}. 
\textit{black} ($x^{(14)}$) is the proportion of black color presented in the frame.

For dynamic visual features,
\textit{grayCor} ($x^{(1)}$) is calculated by one minus the pixel-level correlation between
two consecutive images that are converted into grayscale \citep{zhang2020frontiers}. 
\textit{colDis} ($x^{(2)}$) is the Euclidean distance between the color histograms for red, green, and blue between consecutive video frames \cite{li2019video}. 
Apply a pre-trained deep
neural network ResNet \citep{he2016deep} to the video frames separately and extract the second last
layer feature vectors, which often capture high-level visual information about colors, shapes,
objects, and textures. \textit{resSim} ($x^{(3)}$) is the cosine similarity between the vectors as a measure of
visual dynamics.

For the sound features, \textit{energy} ($x^{(4)}$), \textit{spectral} ($x^{(5)}$) and \textit{entropy}($x^{(9)}$) are the sound energy,  pitch and
entropy for all sounds.  \textit{eneVoc} ($x^{(6)}$), \textit{specVoc} ($x^{(7)}$) and \textit{entropyVoc} ($x^{(8)}$) are the sound energy,  pitch and
entropy for
vocals only.

\end{appendix}

\end{document}